\documentclass[10pt]{amsart}
\pdfoutput=1
\usepackage {amssymb,latexsym}
\usepackage [all]{xy}
\usepackage{nohyperref}
\usepackage[pageref]{backref}
\usepackage{caption}
\usepackage[algo2e]{algorithm2e}
\usepackage{listings}
\usepackage{graphicx}
\newtheorem{theorem}{Theorem}[section]

\newtheorem {proposition}[theorem]{Proposition}

\theoremstyle{definition}
\newtheorem{definition}[theorem]{Definition}
\newtheorem{example}[theorem]{Example}
\newtheorem{examples}[theorem]{Examples}

\theoremstyle{remark}
\newtheorem{remark}[theorem]{Remark}
\newtheorem{remarks}[theorem]{Remarks}

\numberwithin {equation}{section}

\def\ra{{\rightarrow}}

\def\r{{\mathbb{R}}}
\begin{document}
\title{Local Graph Embeddings Based on Neighbors Degree Frequency of Nodes}

\author{Vahid Shirbisheh}
\email{shirbisheh@gmail.com}

\date{\today}
\keywords{Graph Embeddings, Vector Representation of Nodes, Matrix Representation of Nodes, Local Features of Graphs, Centrality Measures, Closeness Centrality, PageRank, Parametric Centrality Measures, Dynamic Graphs, Inductive Graph Deep Learning Models, Graph Machine Learning, Graph Feature Engineering.}

\begin{abstract} 
	We propose a local-to-global strategy for graph machine learning and network analysis by defining certain local features and vector representations of nodes and then using them to learn globally defined metrics and properties of the nodes by means of deep neural networks. By extending the notion of the degree of a node via Breath-First Search, a general family of {\bf parametric centrality functions} is defined which are able to reveal the importance of nodes. We introduce the {\bf neighbors degree frequency (NDF)}, as a locally defined embedding of nodes of undirected graphs into euclidean spaces. This gives rise to a vectorized labeling of nodes which encodes the structure of local neighborhoods of nodes and can be used for graph isomorphism testing. We add flexibility to our construction so that it can handle dynamic graphs as well. Afterwards, the Breadth-First Search is used to extend NDF vector representations into two different matrix representations of nodes which contain higher order information about the neighborhoods of nodes. Our matrix representations of nodes provide us with a new way of visualizing the shape of the neighborhood of a node. Furthermore, we use these matrix representations to obtain feature vectors, which are suitable for typical deep learning algorithms. To demonstrate these node embeddings actually contain some information about the nodes, in a series of examples, we show that PageRank and closeness centrality can be learned by applying deep learning to these local features. Our constructions are flexible enough to handle evolving graphs. In fact, after training a model on a graph and then modifying the graph slightly, one can still use the model on the new graph. Moreover our methods are inductive, meaning that one can train a model based on a known graph and then apply the trained model on another graph with a similar structure. Finally, we explain how to adapt our constructions for directed graphs. We propose several methods of feature engineering based on higher order NDF embeddings and test their performance in several specific machine learning algorithms on graphs.
\end{abstract}
\maketitle

\tableofcontents


\section{Introduction}
\label{sec:intro}

A fundamental challenge in graph machine learning is that graphs in real world applications (aka networks) have a {\it large} and {\it varying} number of nodes (vertices) and edges (links or relationships). Therefore machine learning models on evolving graphs require a considerable time for feature extraction, training and inference. This limits their practical usage in real time applications concerning dynamic graphs. On the other hand, many graph algorithms and machine learning algorithms on graphs rely on operations involving the adjacency matrix or the Laplacian matrix of graphs which are computationally very expensive on large graphs. Time complexity problem also arises in the computation of various globally defined centrality measures, such as closeness centrality. Therefore, for machine learning on large dynamic graphs and generally for large network analytics, one wishes for two things: feature engineering based on local structures of graphs and inductive algorithms suitable for dealing with dynamic graphs.  

In this work, we try to fulfill both these wishes by proposing various locally defined embeddings of nodes into euclidean vector spaces and then designing inductive algorithms to learn globally defined features and properties of nodes of graphs. We only use the adjacency lists of graphs and never get involved with complexity of matrix computations on adjacency matrices or Laplacian matrices of graphs. Among other benefits, this makes our methods more suitable for parallel and distributed computations, although we do not propose any parallel algorithm in this article.

The adjective ``local'' in our nomenclature has a very specific and concrete meaning. It means we use the Breadth-First Search (BFS) algorithm to traverse around the neighborhood of a node and collect the local information. Therefore, implementations of our methods may include a natural number parameter to determine how many layers of the BFS should be performed. Another basic concept that we use is the degree of nodes. Only by combining these two very basic notions, we are able to design several ruled based node representations that contain a great deal of information about the local structure of the neighborhoods of nodes. 

These representations of nodes have several theoretical and machine learning applications. They can be used in graph isomorphism testing. They gives rise to a local approximation of neighborhoods of nodes by simple trees. They provide a simple way of visualizing the neighborhoods of nodes as color maps. They can be considered as feature vectors to be fed into deep learning algorithms to learn global properties of nodes. Finally we show how one can perform feature engineering on them to extract richer feature vectors suitable for specific downstream tasks. These applications are just the tip of iceberg and we expect our methods of graph representation find more applications in other areas of machine learning on graphs and network analysis. 

We begin with exploiting the BFS and the degrees of nodes to define a general family of parametric centrality functions on undirected graphs in Section \ref{sec:parametric-centrality}. These centrality functions are defined as certain summations. We demonstrate the prospect of these parametric centrality functions on very small graphs appearing in network science. However, they are not our primary objective, but they serve us mainly as the simple prototypes and introductory constructions that guide us to more sophisticated methods and constructions. 

To achieve our aforementioned goals in machine learning on huge dynamic graphs, we propose a decomposition of the degrees of nodes into vectors in Section \ref{sec:ndf_vec}. Our way of decomposing the degrees of nodes, which is the main source of our various vector representations, is the counting of occurrences of different degrees among the neighbors of nodes, hence the name {\bf neighbors degree frequency} or shortly {\bf NDF}. It is a vector representation of nodes, namely a mapping that maps a node to its NDF vector inside a euclidean vector space. In other words, it is a node embedding. The NDF vector of a node contains not only the information of its degree, but also the information of the degrees of its neighbors. This two layers of information naturally lays out a perspective of the neighborhood of the node and it can be interpreted and utilized in several ways. For instance, we show that the NDF vector of a node amounts to an approximation of its neighborhood by means of a tree of height 2. We also interpret the NDF vectors as a labeling of nodes that needs to be preserved by isomorphisms of graphs. In this way, NDF vectors and their higher order generalizations find application in graph isomorphism testing. The details of these topics are presented in Subsections \ref{subsec:isom} and \ref{subsec:rndfc-equivalence}.

Whenever there exit a larger number of different degrees in the graph, bookkeeping of the frequency of degrees among neighbors of nodes can increase the space complexity of algorithms. On the other hand, the set of occurring degrees in a dynamic graph can change with every node/edge added or removed from the graph. Therefore we propose a form of aggregation in the degree frequency counting in Section \ref{sec:dynamic}. This gives rise to a diverse family of dynamic NDF vector representations which are capable of both reducing the dimension of the NDF vectors and handling changes in dynamic graphs easily. Dynamic NDF vectors depend on how we partition the set of possible degrees in a graph, so we discuss several practical and intuitively reasonable ways of implementing these partitions.   

Two layers of information about each node does not seem enough for employing NDF vectors as feature vectors to fed into deep learning algorithms. Therefore we extend the definition of NDF vector representations into two different higher order NDF matrix representations in Section \ref{sec:ndf_mat}. Every row of these matrix representations of nodes is an aggregation of NDF vectors of layers of the BFS ({\bf NDFC} matrices) or the degree frequency of layers ({\bf CDF} matrices). Besides using these matrix representations as feature vectors in deep learning algorithms, they have other applications. Like NDF vectors, higher order NDF matrices can be used in graph isomorphism testing. They perform better than color refinement algorithm in several examples. One can also use these matrices to visualize the neighborhoods of nodes in the form of color maps. 

In Section \ref{sec:pagerank}, we flatten NDFC matrices as one dimensional arrays and consider them as feature vectors associated to nodes of graphs. We feed these feature vectors into a simple feedforward neural network model and obtain a supervised deep learning algorithm to learn and predict PageRanks of nodes in undirected graphs. We test some other ways of feature engineering and a special form of convolutional neural networks too. We also show how our method for learning PageRank is applicable for predicting PageRank in dynamic and evolving graphs and even in unseen graph. 

We delve further into feature engineering using our higher order NDF matrix representations in Section \ref{sec:aggregations-closeness}. We consider the un-normalized version of CDF matrices and perform a specific type of aggregation on them. The result is a set of feature vectors suitable for learning the closeness centrality of nodes by means of  feedforward neural networks. Our method in this section are also inductive and easily applicable to dynamic graphs. 

PageRanks of nodes are usually computed in directed graphs, specifically in the network of webpages and hyperlinks in the internet. Therefore one would like to extend the method developed for learning PageRanks of undirected graphs, in Section \ref{sec:pagerank}, to directed graphs. We explain how NDF embeddings and their higher order generalizations can be defined on directed graphs in Section \ref{sec:digraphs}. We also propose a minor change in the definition of NDFC matrices which is needed to take into account the adverse effect of certain directed edges in PageRank. Afterwards, we apply the directed version of NDFC matrices for learning the PageRanks of directed graphs. 

Our methods and constructions depend only on very basic concepts in graph theory. On the other hand, they have shown several early successes in learning PageRank and closeness centrality by means of deep learning. Therefore we believe the theory of NDF embeddings and its higher order generalizations will find a wide range of applications in graph theory, network science, machine learning on graphs and related areas. Besides, it appears that our constructions are new in graph theory and computer science, and so there are so many details that need a through investigation to mature the subject. 

\noindent{\large {\bf Related works. }} The author is new in network science, graph representation learning and machine learning on graphs and does not claim expertise in these research areas. That being said, we believe almost all constructions and methods developed in this article are new and have not appeared in the work of other researchers. However, near to the completion of this work, we found out that a very special case of what we call dynamic NDF vectors appeared before in the work of Ai et al \cite{AZCSL}. They introduced a centrality measure based on the distribution of various degrees among the neighbors of nodes. They use iterated logarithm to associate a vector to each node. In our terminology, this vector is the dynamic NDF vector associated to the list of starting points $\{1, 2, 3, 5, 17, \cdots, 2^n+1, 2^{2^n} + 1 \}$, see \cite{AZCSL} and Section \ref{sec:dynamic} for details and precise comparison. We also notice that what we call the NDF-equivalence among nodes of a graph coincides with the coloring of the nodes obtained after two rounds of color refinement, see Subsection \ref{subsec:isom} for details. 

In the following sections, all graphs are assumed undirected and simple, with no loops and multiple edges, except in Section \ref{sec:digraphs} where we deal with directed graphs. We denote graphs using the notation $G=(V, E)$, meaning $G$ is the graph, $V$ is the set of its vertices that we call them {\it nodes}, and $E$ is the set of {\it edges} of the graph. 

\section{Parametric Centrality Functions}
\label{sec:parametric-centrality}

The Breadth-First Search algorithm in graphs has a mathematical interpretation in terms of a metric structure on the set of nodes. Here, we review this point of view and use it to associate a certain sequence of numbers to each node of a graph. We show how these sequences of numbers may be used to construct parametric centrality measures applicable in network science.   

In mathematics, a {\it metric space} is a set $X$ equipped with a {\it metric} function $d: X\times X\ra [0,+\infty)$ such that for all $x, y, z\in X$, we have 
\begin{itemize}
	\item [(i)] $d(x,y) = d(y,x)$ {\it (symmetry)}.
	\item [(ii)] $d(x,y) \geq 0$, and the equality holds if and only if $x = y$ {\it (positivity)}. 
	\item [(iii)]$d(x,z) \leq d(x,y) + d(y,z)$ {\it (triangle inequality)}.
\end{itemize}
A metric space is usually denoted by the pair $(X,d)$. 

The {\it natural metric space structure} on the set of nodes of a graph $G=(V, E)$ is defined by the distance of every two nodes as the metric function. We recall that the distance $d(u, v)$ of two nodes $v, u\in V$ is defined as the number of edges in any shortest path connecting $u$ and $v$. In order to this definition makes sense, here, we assume that $G$ is connected. With this assumption, it is easy to verify the axioms of a metric space for $(V, d)$. By adding $+\infty$ to the target space of the metric function and defining the distance of two nodes in two distinct components to be $+\infty$, one can extend this definition to all undirected graphs, possibly disconnected. We use this extended definition when the graph under consideration is not connected. Although $d$ is a metric on the set $V$, we may interchangeably call it the natural metric on the graph $G$. One notes that this metric structure on $V$ induces a {\it discrete topology} on $V$, meaning every subset of $V$ is both open and closed in this topology. Therefore we drop the terms {\it open and closed} in our discussion.

The notion of a metric function extends the euclidean distance in $\r^n$. Thus similar to  intervals, disks and balls in $\r$, $\r^2$ and $\r^3$, respectively, we can define corresponding notions in any metric space as well. However, since most graph drawings are in two dimensions, we borrow the terms {\it "disk"} and {\it "circle"} from $\r^2$ to use in our terminology. 

Given an undirected graph $G=(V,E)$, for every node $v\in V$ and integers $k = 0,1,2,\cdots$, we define the {\bf disk} of radius $k$ centered at $v$ by 
\begin{equation*}
	\label{eq:disk}
	D_k(v) = \{u \in V; d(v, u) \leq k \},
\end{equation*}
and similarly, the {\bf circle} of radius $k$ centered at $v$ is defined by 
\begin{equation*}
	\label{eq:circle}
	C_k(v) = \{u \in V; d(v, u) = k \}.
\end{equation*}
Next, computing the cardinality of the above circles defines a sequence of functions $S=\{s_k \}_{k\geq 0}$ on $V$, i.e. 
\begin{equation*}
	\label{eq:seq_sizes}
	s_k(v) = |C_k(v)|, \quad \forall v \in V.
\end{equation*}
Here $S$ stands for the word ``{\it size}'' and let us call this sequence the {\bf sizes of circles} around a node.
Now, we are ready to define the general family of parametric functions on graphs.

\begin{definition}
	Let $G=(V, E)$ be an undirected graph. To any given sequence $\Lambda=\{\lambda_k \}_{k\geq 0}$ of real numbers, we associate a parametric function $\varphi_\Lambda : V \ra \r$ defined by
	\[
	\Phi_\Lambda(v) = \sum_{k = 0}^{\infty} \lambda_k s_k(v), \forall v\in V.
	\]
	We call this function the {\bf parametric centrality function} associated with $\Lambda$. 
\end{definition} 
When $\Lambda$ is known, we may drop it from our notation. We note the followings:
\begin{remarks}
	\begin{itemize}
		\item [(i)] Since we are dealing with finite graphs, the number of non-zero functions in the sequence $\{s_k \}_{k\geq 0}$ is bounded by the diameter of the graph (plus one). Therefore the above series is always finite and well defined.
		\item [(ii)] By choosing $\lambda_1 = 1$ and $\lambda_k = 0$ otherwise, we get the degree centrality (up to a scaling factor).
		\item [(iii)] In practice, we let only first few parameters in $\Lambda$ to be non-zero. Thus these parametric functions are indeed locally defined features on graph nodes. 
		\item [(iv)] By carefully choosing $\Lambda$ according to the downstream applications, $\Phi_\Lambda$ can be regarded as a {\it parametric centrality measure} of graph nodes, see Examples \ref{exam:countries_1}, \ref{exam:main_parametric} and \ref{exam:three_classic} below.
		\item [(v)] Even the sequence $\{s_k\}_{k\geq 0}$ contains some information about the local structure of the graph around its nodes. For instance, comparing the value $s_2$ against another local feature of the graph gives rise to a measure of the inter-connectivity of neighborhoods of nodes, see Proposition \ref{prop:local-tree-condition}.
		\item [(vi)] The elements of the sequence $\{s_k\}$ increase in higher density neighborhoods. Alternatively, one can use the sequence $\{\frac{1}{s_k}\}$ (or better $\{\frac{1}{s_k^\alpha + \varepsilon}\}$ for some $\alpha, \varepsilon > 0$) in lieu of $\{s_k\}$ to define a family of parametric functions such that they decrease on nodes in higher density neighborhoods.
		\item [(vii)] When nodes or edges are weighted, one can incorporate those weights in the definition of the sequence $\{s_k\}$ to obtain more meaningful centrality measures on graph nodes. 
	\end{itemize}
\end{remarks}
In the following examples, we expand on the prospect of the parametric centrality functions in network science.  Here, we consider a relatively simple way of defining the sequence $\Lambda$ of parameters. For some $0<p<1$, we set $\lambda_k= p^{k - 1}$ for $k > 0$ and $\lambda_0=0$. For the sake of simplicity, let us call the parametric centrality function associated to this $\Lambda$ the {\bf $p$-centrality function} and denote it simply by $\Phi_p$. It is interesting to note that this very simple family of parametric centrality measures can mimic (to some extent) some of the well known centrality measures. 
\begin{example}
	\label{exam:countries_1}
	Let $G = (V, E)$ be the network of neighboring countries in a continent, i.e. the nodes are countries and two countries are connected if they share a land border. It is a reasonable assumption that the financial benefits of mutual trades between two countries decrease by some factor $0<p<1$ for every other country lying in a shortest path connecting them. Therefore the total economical benefits of continental trade for each country depends on both the distance of countries and the number of other countries in each circle around the country. In mathematical terms, for a given country $X$, the financial benefits of trade between $X$ and countries with distance $k$ from $X$ is $p^{k-1} s_k(X)$, because there are exactly $k - 1$ countries in any shortest path connecting them. The total benefits of continental trade is the summation over all distances $k$, i.e. $\sum_{k=1}^\infty p^{k-1} s_k(X)$ and this is exactly $\Phi_p(X)$. Therefore the $p$-centrality function $\Phi_p$ models the total financial benefits of continental trade for each country according to the graph structure of this network (and our assumptions!).
\end{example} 

\begin{figure}
	\centering
	\includegraphics[width=350px]{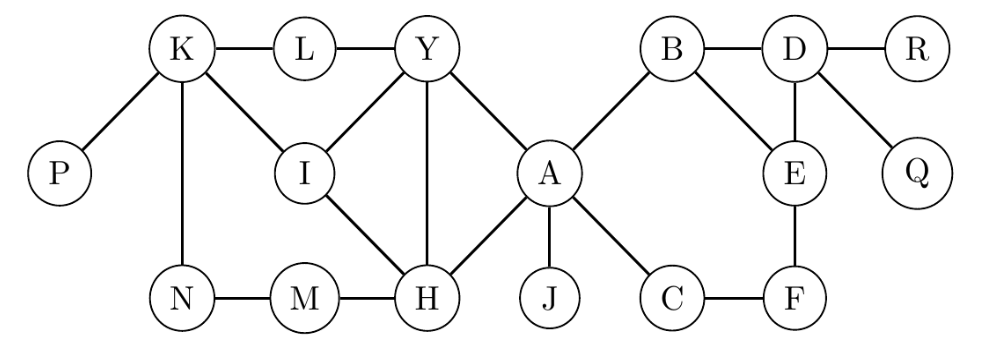}
	\caption{The graph in Examples \ref{exam:main_parametric}, \ref{exam:mndf}, \ref{exam:tree_approx}, \ref{exam:g1_g2_ndf}(iii) and \ref{exam:main_graph_dndf}. } \label{fig:main-graph}
\end{figure}

We use the graph shown in Figure \ref{fig:main-graph} in several occasions in this paper to explain various concepts and constructions. It was originally designed to explain the content of Section \ref{sec:ndf_vec}. However, it appears to be explanatory for prospects of parametric centrality functions as well.
\begin{example}
	\label{exam:main_parametric}
	Let $G = (V, E)$ be the undirected graph shown in Figure \ref{fig:main-graph}. We heuristically set $p=0.6$ and compute $p$-centrality function $\Phi_p$ on nodes of $G$. The result is divided by the scaling factor 21.15 in order to be comparable to closeness centrality. In Table \ref{table:main_parametric1}, we compare the sorted dictionary of $p$-centrality function with the sorted dictionaries of closeness, eigenvector and betweenness centrality measures (computed using the Python package NetworkX \cite{HSS}). In each case, we arrange the result in decreasing order to show that the order of nodes according to $p$-centrality function is similar to the order with respect to closeness centrality. One observes that those nodes which are not in the same order happen to have very close $p$-centrality as well as closeness centrality. The order of values of $p$-centrality function shows less similarity with the order of values of eigenvector centrality. One can also spot more disarray while comparing $p$-centrality function with betweenness centrality.  In Table \ref{table:main_parametric2}, we compare the values of $p$-centrality function (again divided by the scaling factor 21.15) with the values of closeness centrality. It shows that $p$-centrality function fluctuates around closeness centrality with an average error of about 3.938 percent. 
	
	The diameter, the maximum distance between two nodes, of graph $G$ is 7. So for computing $p$-centrality function, we only need to consider circles of radius at most 7. If we restrict our computations to circles with strictly smaller radiuses, we get less accurate results. For instance, the average percentage of differences between closeness and $p$-centrality (the last row of Table \ref{table:main_parametric2}) will be 4.149 percent (respectively 4.411, 5.868, 13.401, 30.941 and 65.599 percent) if we use circles of radius at most 6 (respectively 5, 4, 3, 2 and 1). This suggests that the higher maximum radius, the more precise estimation of closeness centrality. It also shows that by restricting our computations to circles of radius at most 5, we lose about 0.473 percentage of accuracy, but it decreases the computational burden. Therefore, in larger graphs, the maximum radius of circles involved in the computation can be a notable hyperparameter to establish a balance between precision and time complexity of the computation. Finally we must mention that in a low maximum radius computation of $p$-centrality, one should consider a different scaling factor to achieve a better estimation. For example, in the computation of $p$-centrality function of $G$ involving circles of at most radius 2 and with the scaling factor 15.7, the average percentage of differences drops to 14.631 percent (from 30.941 percent with the scaling factor 21.15).
	\begin{table}
		\begin{tabular}{|c|c||c|c||c|c||c|c|}\hline
			\multicolumn{2}{|c||}{Closeness} & \multicolumn{2}{c||}{$p$-Centrality} & \multicolumn{2}{c||}{Eigenvector} & \multicolumn{2}{c|}{Betweenness}    \\
			\multicolumn{2}{|c||}{Centrality} & \multicolumn{2}{c||}{Function} & \multicolumn{2}{c||}{Centrality} & \multicolumn{2}{c|}{Centrality}    \\\hline
			A & 0.485 & A & 0.485 & Y & 0.444 & A & 0.613 \\\hline
			H & 0.444 & H & 0.448 & A & 0.437 & B & 0.346 \\\hline
			Y & 0.432 & Y & 0.437 & H & 0.436 & H & 0.261 \\\hline
			B & 0.410 & B & 0.418 & I & 0.347 & D & 0.242 \\\hline
			I & 0.364 & I & 0.378 & K & 0.236 & Y & 0.229 \\\hline
			C & 0.356 & D & 0.361 & B & 0.224 & K & 0.156 \\\hline
			L & 0.348 & E & 0.353 & L & 0.211 & I & 0.128 \\\hline
			M & 0.340 & C & 0.351 & M & 0.175 & M & 0.086 \\\hline
			D, E, J & 0.333 & K & 0.348 & C & 0.165 & C & 0.083 \\\hline
			F & 0.314 & L & 0.347 & E & 0.143 & L & 0.061 \\\hline
			K & 0.308 & M & 0.336 & D & 0.141 & E & 0.046 \\\hline
			N & 0.291 & F & 0.311 & J & 0.136 & F, N & 0.021 \\\hline
			Q, R & 0.254 & J & 0.310 & N & 0.128 & J, P, Q, R & 0.0 \\\hline
			P & 0.239 & N & 0.298 & F & 0.096 &  &  \\\hline
			&  & Q, R & 0.235 & P & 0.073 &  &  \\\hline
			&  & P & 0.227 & Q, R & 0.044 &  &  \\\hline
		\end{tabular}
		
		\caption{\label{table:main_parametric1} The sorted lists of closeness centrality, $p$-centrality function, eigenvector centrality and betweenness centrality of the graph shown in Figure \ref{fig:main-graph}. Each value has been rounded to three decimal places. }
	\end{table}
	\begin{table}
		\begin{tabular}{|c|c|c|c|}\hline
			Node & Closeness & $p$-Centrality & Percentage of \\
			&  Centrality & Function & Difference \\\hline
			A & 0.485 & 0.485 & 0.014\% \\\hline
			H & 0.444 & 0.448 & 0.766\% \\\hline
			B & 0.410 & 0.418 & 1.898\% \\\hline
			Y & 0.432 & 0.437 & 0.941\% \\\hline
			I & 0.364 & 0.378 & 3.852\% \\\hline
			C & 0.356 & 0.351 & 1.149\% \\\hline
			L & 0.348 & 0.347 & 0.127\% \\\hline
			E & 0.333 & 0.353 & 5.886\% \\\hline
			D & 0.333 & 0.361 & 8.155\% \\\hline
			J & 0.333 & 0.310 & 7.041\% \\\hline
			M & 0.340 & 0.336 & 1.289\% \\\hline
			F & 0.314 & 0.311 & 0.759\% \\\hline
			K & 0.308 & 0.348 & 12.979\% \\\hline
			N & 0.291 & 0.298 & 2.594\% \\\hline
			Q & 0.254 & 0.235 & 7.381\% \\\hline
			R & 0.254 & 0.235 & 7.381\% \\\hline
			P & 0.239 & 0.227 & 4.739\% \\\hline
			\multicolumn{3}{|c|}{Average percentage of Differences} & \multicolumn{1}{c|}{3.938\%}   \\\hline
			
		\end{tabular}
		
		\caption{\label{table:main_parametric2} The value comparison of closeness centrality and $p$-centrality function of nodes in the graph shown in Figure \ref{fig:main-graph}. Each value has been rounded to three decimal places. The last column is computed using the formula $\frac{100|C - \Phi_p|}{C}$, where $C$ is the closeness centrality and $\Phi_p$ is $p$-centrality function. We rounded the result after applying the formula on raw values. This explains why node $A$ has equal closeness centrality and $p$-centrality in the table, but the percentage of the difference is slightly bigger than zero. }
	\end{table}
	
\end{example}
In the next example, we again compare certain $p$-centrality functions to closeness centrality on three famous graphs in network science and we obtain similar results as the above example.

\begin{examples} 
	\label{exam:three_classic}
	We use the Python library NetworkX, \cite{HSS}, for generating the following graphs and for computing the closeness centrality of their nodes. To compute $p$-centrality functions for all three cases, we use circles of radius at most 5. We refer the reader to \cite{HSS} for corresponding references for each graph.
	\begin{itemize}
		\item[(i)] For Zachary’s Karate Club graph, we choose $p=0.6$. Then by dividing the $p$-centrality function by the scaling factor $42$, we get an estimation of the closeness centrality function by an average error of about 2 percent.  
		\item[(ii)] For the graph of Florentine families, we set $p=0.6$ and let the scaling factor be $17.96$. Then, the average differences between $p$-centrality function and closeness centrality is about 3.416 percent.
		\item[(iii)] For the graph of co-occurrence of characters in the novel Les Miserables, we set $p=0.565$ and let the scaling factor be $86.8$. In this case, the average differences between $p$-centrality function and closeness centrality is about 5.46 percent.
	\end{itemize}
\end{examples}

A word of caution is in order regarding the above examples. Our primary goal is not to prove that $p$-centrality function is similar to (or approximates) other centrality measures. These examples only suggest that $p$-centrality functions and more generally parametric centrality functions can equally be considered as certain measures to evaluate the importance of nodes in a graph. Furthermore, since we compute parametric centrality functions for each node based on its local neighborhood, the computation of these functions are more practical and computationally cheaper than other centrality measures, which are computed based on the information of the whole graph. Therefore, in certain applications in network science, they may offer cheaper and more flexible substitutes for well-known centrality measures. 	

In certain problems in network science, the factor $p$ can be the probability of some interaction (transmission of information or diseases, etc) between nodes in a network. In these cases, one might need to define $\lambda_k = p^k$ in lieu of $p^{k - 1}$ for $k>0$ and the appropriate parametric centrality function would be $\sum_{k=1}^\infty p^k s_k$. The bottom line is our definition has enough flexibility to be useful in many situations occurring in real world problems in network science. 
  
Counting the number of nodes in each circle (layer of the BFS) around graph nodes reveals only a small portion of the local structure of graphs. To extract more refined local information of each node, we continue with constructing a vector representation of nodes in the next section.

\quad
\section{The Neighbors Degree Frequency (NDF) Vector Representation of Nodes}
\label{sec:ndf_vec}
In this section, we construct the simple version of our graph embedding and discuss its interpretations, usages and shortcomings. This leads us to more mature versions of our construction in upcoming sections.
\begin{definition}
	\label{def:ndf_vanilla}
	Let $G = (V, E)$ be an undirected graph and let $m$ be the maximum degree of nodes in $G$. For every node $v \in V$, we define a vector $vndf(v) = (vndf(v)_1,\cdots,vndf(v)_m)\in\mathbb{R}^m$ by setting $vndf(v)_i$ to be the number of neighbors of $v$ of degree $i$ for all $i = 1, \cdots, m$. We call the function $vndf:V\rightarrow \mathbb{R}^m$ the {\bf vanilla neighbors degree frequency embedding} of nodes of $G$, and briefly {\bf vanilla NDF}. 
\end{definition}
  
The main point of our definition is to count the frequency of various degrees among neighbors of a node, hence the name {\bf Neighbors Degree Frequency} and the acronym {\bf NDF}. We basically decompose the degree of a node into a vector consisting of degree frequencies of its neighbors, as one computes $deg(v) = s_1(v)=\sum_{i = 1}^{m} vndf(v)_i$. Using $vndf(v)$, one can also find an upper bound for the size of the circle of radius 2 around $v$, i.e. $s_2(v)$, see Proposition \ref{prop:local-tree-condition}.

The vanilla version of neighbors degree frequency embedding has three obvious shortcomings: First it is wasteful, because there might exist numbers between 1 and $m - 1$ that never occur as a degree of a node in $G$. For example, consider a star shape graph with 10 nodes. It has 9 nodes with degree 1 that are adjacent to the central node of degree 9. The vanilla NDF requires 9 euclidean dimensions, but only the first and the last entries are used in the image of the mapping $vndf$. To remedy this issue, we propose a minimal version of NDF embedding in Definition \ref{def:ndf_minimal} below.

The second problem occurs when networks grow. Since the nodes with higher degrees usually attract most of the new links (for instance, because of preferential attachment process), there is a good chance that the maximum degree $m$ of graph nodes increases in growing graphs, and so a fixed $m$ wouldn't be sufficient for the above definition any more. This issue on small dynamic graphs has an easy solution: One can count all nodes of degree greater or equal to $m$ as degree $m$. 

The third problem is the {\it curse of dimensionality} for graphs having nodes with very high degrees. For instance, graphs appearing in social networks may contain nodes with millions of links. Clearly, a vector representation in a euclidean space of several million dimensions is not a practical embedding!. For the case of big dynamic graphs, we propose a dynamic version of NDF in Section \ref{sec:dynamic}. Our solution for big dynamic graph is perfectly applicable to the above cases too. However, since the vanilla and minimal NDF vector representations have several conceptual and theoretical importance, we devote a few more paragraphs to them.

\begin{definition}
	\label{def:ndf_minimal}
	Let $G=(V,E)$ be an undirected graph. Let $\{d_1,\cdots, d_m\}$ be the ascending list of all occurring degrees of nodes in $G$, which we call it the {\bf degrees list} of the graph $G$. The {\bf minimal neighbors degree frequency embedding} of nodes of $G$ is defined as
	\[
	mndf(v) = (mndf(v)_1,\cdots,mndf(v)_m)\in\mathbb{R}^m, \qquad \forall v \in V,
	\]  
	where $mndf(v)_i$ is the number of neighbors of $v$ of degree $d_i$, for all $i=1,\cdots, m$.
\end{definition} 
The minimal NDF is essentially the same as vanilla NDF except that we have omitted never occurring degrees while recording the frequency of degrees of neighbors of each node. The important point we need to emphasize here is that we substitute the role played by the ``maximum degree'' of graph nodes with ``degrees list'' of the graph. This paves the way for more variations of NDF by upgrading the concept of ``degrees list'' into ``intervals list''. This is our main pivot to define more practical versions of NDF for huge dynamic graphs in Section \ref{sec:dynamic}. In the following example, we examine some aspects of vanilla (and minimal) NDF:
	
\begin{example}
	\label{exam:mndf}
	Let $G=(V, E)$ be the graph shown in Figure \ref{fig:main-graph}. The ascending list of degrees in $G$ is $\{1,2,3,4,5\} $. Hence the minimal NDF and the vanilla NDF are the same and they map the nodes of $G$ into $\mathbb{R}^5$. The NDF vector representations of nodes of $G$ are shown in Table \ref{table:g_mndf}. One immediately notices the followings:
	\begin{itemize}
		\item[(i)] Nodes $B, E, I$ are all of degree 3, but they have very different NDF vector representations. This shows how NDF is much more detailed criterion than degree centrality in classifying nodes. 
		\item[(ii)] Nodes $Y$ and $H$ have the same NDF vector representations, because their neighbors have similar degree frequency. But the circles of radius 2 around them do not have the same size. In fact, one easily checks that $s_2(Y)=5$ and $s_2(H)=6$. This hints that in order to define a more faithful representation of nodes, we cannot content ourselves to just NDF vectors. Using the sequence $\{s_k\}_{k\geq 2}$ is one way to gain more in depth knowledge of the neighborhood topology. We shall also propose two extensions of NDF vectors to circles of higher radius around a node to extract more information about the bigger neighborhoods of nodes in Section \ref{sec:ndf_mat}. 
		\item[(iii)] Nodes $Q$ and $R$ cannot be distinguished from each other by all these measures. They essentially have the same statistics of degree frequency and circles sizes at all levels. In this case, it is clear that the permutation which only switches these two nodes is an automorphism on $G$ (an isomorphism from $G$ into itself). Therefore it is expected that they have the same NDF vector representation, see also Example \ref{exam:main_graph_auto}.
	\end{itemize}
		
	\begin{table}
		\begin{tabular}{|c|c|}\hline
			Node(s) & NDF Vector \\\hline
			A & (1, 1, 1, 2, 0)  \\\hline
			D & (2, 0, 2, 0, 0)  \\\hline
			K & (1, 2, 1, 0, 0)  \\\hline
			Y, H & (0, 1, 1, 1, 1)  \\\hline
			B & (0, 0, 1, 1, 1)  \\\hline
			E & (0, 1, 1, 1, 0)  \\\hline
			I & (0, 0, 0, 3, 0)  \\\hline
			C & (0, 1, 0, 0, 1) \\\hline
			F & (0, 1, 1, 0, 0) \\\hline
			L & (0, 0, 0, 2, 0) \\\hline
			N, M & (0, 1, 0, 1, 0)  \\\hline
			J & (0, 0, 0, 0, 1) \\\hline
			P, Q, R & (0, 0, 0, 1, 0)  \\\hline
		\end{tabular}
		
		\caption{\label{table:g_mndf} The vanilla (and minimal) NDF embedding of nodes of graph $G$, shown in Figure \ref{fig:main-graph} and explained in Examples \ref{exam:mndf}, \ref{exam:tree_approx} and \ref{exam:main_graph_auto}. }
	\end{table}
\end{example}
 
\subsection{Tree approximations of neighborhoods of nodes}
\label{subsec:tree-approx}\quad \\

One might interpret vanilla and minimal NDF vector representations of nodes as an extension of degree of nodes into an embedding of nodes into vector spaces. However there is a more explanatory picture of NDF which is the subject of this subsection. Here we explain how NDF gives rise to an approximation of local neighborhoods of nodes by means of trees. The following proposition with minor modifications is valid for the minimal NDF too. For the sake of simplicity, we state it for the vanilla NDF only. First, we need some notations.

Let $G=(V, E)$ be an undirected graph. We recall that the {\it ego subgraph} of a graph $G$ centered at $v\in V$ of radius $r$, which we denote it by $E_r(v)$, is the induced graph whose set of nodes equals $D_r(v)$, the disk of radius $r$ around $v$, and whose set of edges consists of those edges of $G$ that both ends belong to $D_r(v)$, see \cite{HSS}. 

\begin{proposition}
	\label{prop:local-tree-condition} With the notation of Definition \ref{def:ndf_vanilla}, we define a function $\hat{s}:V\rightarrow \mathbb{N}$ by setting $\hat{s}(v) := \sum_{i = 1}^{m} (i-1)vndf(v)_i$, i.e. $\hat{s}(v)$ is defined as the dot product of $vndf(v)$ and the vector $(0, 1, \cdots, m-1)$. Then, for every $v\in V$, we have the followings:
	\begin{itemize}
		\item [(i)] $s_2(v) \leq \hat{s}(v)$, for all $v \in V$. 
		\item [(ii)] If $E_2(v)$ is a tree, then $s_2(v) = \hat{s}(v)$.
		\item [(iii)] Let $E_2^\prime(v)$ be the graph obtained from $E_2(v)$ by omitting the edges whose both ends lie in $C_2(v)$. Then $s_2(v) = \hat{s}(v)$ if and only if $E_2^\prime(v)$ is a tree.
	\end{itemize}
	
\end{proposition}
\begin{proof}
	Given a node $u\in C_1(v)$, let $S(u)$ denote $C_1(u)-\{v\}$, the set of neighbors of $u$ besides $v$. Then we have 
	\[
	C_2(v) \subseteq \bigcup_{u \in C_1(v) } S(u) = \bigcup_{i=1}^{m}\left( \bigcup_{u \in C_1(v), deg(u)=i } S(u) \right).
	\] 
	Taking the cardinality of both sides of this inclusion proves (i). 
	
	Assertion (ii)  follows from (iii). To prove (iii), we only need to notice that the equality $s_2(v) = \hat{s}(v)$ is logically equivalent to the presence of the following two conditions: 
	\begin{eqnarray*}
		S(u)&\subseteq& C_2(v), \qquad \forall u\in C_1(v), \\ 
		S(u)\cap S(w) &=& \emptyset, \qquad \forall u\neq w \in C_1(v).
	\end{eqnarray*}
	And this equivalence is clear from the above inclusion.
\end{proof}
  
\begin{example}
	\label{exam:tree_approx}
	In this example, we apply the content of the above proposition to some selected nodes of graph $G$ in Figure \ref{fig:main-graph}. 
	\begin{itemize}
		\item [(i)] One computes $s_2(A) = 6$ and $\hat{s}(A) = 9$. This is because $E_2^\prime(A)$ is not a tree; first $S(Y)$ and $S(H)$ are not subsets of $C_2(A)$, secondly $S(Y)\cap S(H) = \{I\}$.
		\item [(ii)] We observe that $E_2(P)$ is a tree and we have $s_2(P) = \hat{s}(P) = 3$.
		\item [(iii)] We have $s_2(F) = \hat{s}(F) = 3$, but the ego subgraph $E_2(F)$ is not a tree, because of the triangle $\Delta BDE$. However, the graph $E_2^\prime(F)$ is a tree (after omitting the edge (B, D) from $E_2(F)$). 
	\end{itemize}
\end{example}

The insight gained in Proposition \ref{prop:local-tree-condition} leads us to the following concept of local approximation:
\begin{remark}
	\label{rem:tree-approx}
	For every node $v$ in a graph $G=(V, E)$, we define a rooted tree $T_v$ of height 2 as follows: We call the root of the tree $t_v$ and the set of its children is $\{t_u, u\in C_1(v)\}$. In other words the children of the root are indexed by neighbors of $v$. Next, for every neighbor $u$ of $v$, we associate $deg(u) -1$ children $t_{u}^1, \cdots, t_{u}^{deg(u) -1}$ to the node $t_u$. With this setting in place, the intuition behind NDF and Proposition \ref{prop:local-tree-condition} is that the vanilla (and minimal) NDF simplifies (or approximates) every local neighborhood of a node $v$ in a graph $G$ as the tree $T_v$. In fact, $s_2(v) = \hat{s}(v)$ if and only if there is an isomorphism $\varphi: T_v \rightarrow E_2^\prime(v)$ of graphs such that $\varphi(t_v)=v$.  
\end{remark}

As an example for the last statement of the above remark consider node $N$ in the graph $G$ in Figure \ref{fig:main-graph}. Then the mapping $\varphi: T_N \rightarrow E_2^\prime(N)$ defined as $\varphi(t_N)=N$, $\varphi(t_M)=M$, $\varphi(t_K)=K$, $\varphi(t_K^1)=L$, $\varphi(t_K^2)=I$, $\varphi(t_K^3)=P$ and $\varphi(t_M^1)=H$ is an isomorphism of graphs, so $s_2(N)=\hat{s}(N)$.

This approximation is similar to the role played by tangent spaces in geometry of smooth manifolds, because they locally approximate the neighborhood of points on a manifold by means of flat vector spaces. For example consider the tangent line at a point in a smooth curve lying in $\r^2$. At a small neighborhood of that point, the tangent line is a locally good approximation of the curve. Here by local approximation, we mean that there is an invertible function $f$ between an open neighborhood of the point on the curve and an open interval around zero in the tangent line such that both $f$ and $f^{-1}$ are smooth mappings. 

One notes that this analogy is not perfect. For instance, tangent spaces vary ``smoothly'' when points vary smoothly in the manifold. But there is no notion of smoothness in graphs (at least nothing that I am aware of), because graph nodes are essentially discrete creatures. Another important gap in this analogy is the notion of ``dimension'': In smooth (or even topological) manifolds, every point in a connected manifold enjoys a tangent space with the same dimension, which is the dimension of the manifold. Again, we do not have a notion of dimension in graphs or trees yet (to my knowledge). However this analogy helps us to think about higher order approximations of the neighborhood of a node in the graph. In Section \ref{sec:ndf_mat}, we expand the idea of counting the frequency of degrees to higher circles around each node.

\subsection{NDF-equivalence as a requirement for graph isomorphisms}
\label{subsec:isom}\quad \\

Another application of NDF embeddings is in graph isomorphism testing problem. It starts with the observation that if $\theta:G\ra H$ is an isomorphism between two graphs $G$ and $H$, then for every node $v$ in $G$, we must have $vndf(v) = vndf(\theta(v))$. Therefore, vanilla NDF can serve as an easy test for two graphs not being isomorphic. In other words, a necessary condition for two undirected graphs to be isomorphic is that the vanilla NDF embeddings of them should have the same image with the same multiplicity for each element in the image. When studying the group of automorphisms of a graph, vanilla NDF can be used similarly to filter the set of potential permutations of nodes eligible to be an automorphism of the graph, see Part (ii) of Example \ref{exam:g1_g2_ndf} and Example \ref{exam:main_graph_auto}. In order to formulate these ideas in mathematical terms, we need to define two levels of equivalence relations; between nodes and between graphs.
\begin{definition}
	\label{def:ndf-equivalence}
	\begin{itemize}
		\item [(i)] Two nodes in an undirected graph $G=(V, E)$ are called {\bf NDF-equivalent} if their vanilla NDF vectors are the same. The partition of $V$ induced by this equivalence relation is called the {\bf NDF-partition}.
		\item [(ii)] Let $G_1=(V_1, E_1)$ and $G_2=(V_2, E_2)$ be two undirected graphs. The graphs $G_1$ and $G_2$ are called {\bf NDF-equivalent} if the following conditions hold:
		\begin{itemize}
			\item [(a)] The vanilla NDF embedding on them has the same target space, equivalently the maximum degrees of $G_1$ and $G_2$ are equal.
			\item [(b)] The images of the vanilla NDF embedding on $G_1$ and $G_2$ are equal, that is $vndf(G_1) = vndf(G_2)$.
			\item [(c)] For every vector $X\in vndf(G_1)=vndf(G_2)$, the inverse images of $X$ in $V_1$ and $V_2$ have the same number of nodes. That is
			\[
			|\{v\in V_1; vndf(v) = X \}| = |\{v\in V_2; vndf(v) = X \}|.
			\]
		\end{itemize}
	\end{itemize}
\end{definition}

These equivalence relations are defined with respect to vanilla NDF. Similar definitions are possible for other types of NDF embeddings including higher order NDF embeddings. When it is necessary to distinguish between them, we will add an adjective (for example ``vanilla' or ``minimal'', etc) to specify what kind of NDF was used to define the equivalence. It is straightforward to see that two isomorphic graphs are NDF-equivalent. But the converse is not true, even for trees, see Part (i) of the following example:

\begin{example}
	\label{exam:g1_g2_ndf}
	Let $G_1$ and $G_2$ be the following graphs:
	\[
	G_1: 
	\xymatrix{
		*+[o][F-]{1} \ar@{-}[r]&*+[o][F-]{2} \ar@{-}[r]&*+[o][F-]{3} \ar@{-}[r]&*+[o][F-]{4} \ar@{-}[r]&*+[o][F-]{5} \ar@{-}[r] \ar@{-}[d]&*+[o][F-]{6} \ar@{-}[r]&*+[o][F-]{7} \ar@{-}[r]&*+[o][F-]{8} \ar@{-}[r]& *+[o][F-]{9}\\
		& & & &*+[o][F-]{10}& & & &
	}
	\]  
	\[
	G_2: 
	\xymatrix{
		*+[o][F-]{1} \ar@{-}[r]&*+[o][F-]{2} \ar@{-}[r]&*+[o][F-]{3} \ar@{-}[r]&*+[o][F-]{4} \ar@{-}[r] \ar@{-}[d]&*+[o][F-]{5} \ar@{-}[r]&*+[o][F-]{6} \ar@{-}[r]&*+[o][F-]{7} \ar@{-}[r]&*+[o][F-]{8} \ar@{-}[r]& *+[o][F-]{9}\\
		& & &*+[o][F-]{10}& & & & &
	}
	\] 
	
	\begin{itemize}
		\item[(i)] Clearly, $G_1$ and $G_2$ are not isomorphic, but they are NDF-equivalent, see Table \ref{table:g1_g2_ndf}.
		\item[(ii)] If an automorphism of $G_1$ (resp. $G_2$) does not change nodes 1 and 9, it must be the trivial automorphism. On the other hand, since nodes 1 and 9 are NDF-equivalent in graph $G_1$ (resp. $G_2$), every non-trivial automorphism on $G_1$ (resp. $G_2$) should switch nodes 1 and 9. Therefore, with a little of analysis, one observes that there is only one non-trivial automorphism on $G_1$ given with the following permutation:
		\[
		\left(
		\begin{array}{cccccccccc}
			1 & 2 & 3 & 4 & 5 & 6 & 7 & 8 & 9 & 10 \\ 
			9 & 8 & 7 & 6 & 5 & 4 & 3 & 2 & 1 & 10
		\end{array}
		\right).
		\]
		However, since 4 has a unique NDF vector in $G_2$, it is mapped to itself by all automorphisms of $G_2$. Therefore if an automorphism wants to switches 1 and 9, it should map isomorphically the path $(1, 2, 3, 4)$ onto the path $(9, 8, 7, 6, 5, 4)$, which is impossible. Thus, there is no non-trivial automorphism on graph $G_2$. 
	\end{itemize}
	\begin{table}[h]
		\begin{tabular}{|c|c|c|}\hline
		Nodes in $G_1$	& Nodes in $G_2$ & The Associated NDF Vector   \\\hline
		1, 9 & 1, 9 & (0, 1, 0)  \\\hline
		2, 8 & 2, 8 & (1, 1, 0)  \\\hline
		3, 7 & 6, 7 & (0, 2, 0)  \\\hline
		4, 6 & 3, 5 & (0, 1, 1) \\\hline
		5 & 4 & (1, 2, 0) \\\hline
		10 & 10 & (0, 0, 1)\\\hline
		\end{tabular}

		\caption{\label{table:g1_g2_ndf} Comparing the NDF embeddings of graphs $G_1$ and $G_2$ of Example \ref{exam:g1_g2_ndf} }
	\end{table}
\end{example}

\begin{example}
	\label{exam:main_graph_auto}
	In this example, we apply the vanilla NDF of the graph $G$ shown in Figure \ref{fig:main-graph} to prove that the only non-trivial automorphism of $G$ is the permutation switching only nodes $Q$ and $R$. According to the vanilla NDF of nodes of $G$ listed in Table \ref{table:g_mndf}, we have three NDF-equivalence classes with more than one element. Regarding our discussion at the beginning of this subsection, these are the only chances for having a non-trivial automorphism on $G$. Again using NDF vectors, we try to eliminate most of these chances. Nodes $Y$ and $H$ both have only one neighbor of degree 2, i.e. respectively, $L$ and $M$. If an automorphism switches $Y$ and $H$, then it must switch $L$ and $M$ too, which is impossible because $vndf(L)\neq vndf(M)$. Similar arguments show that no automorphism can switch the pairs ($M$, $N$), ($P$, $Q$) and ($P$, $R$). Therefore, the only possibility for an automorphism on graph $G$ is the permutation on the set of nodes which switches nodes $Q$ and $R$. And one easily checks that this permutation is actually an automorphism on $G$.  
\end{example}
Color refinement or 1-dimensional Weisfeiler-Leman algorithm is an iterative algorithm to partition the set of nodes of undirected graphs and is a popular algorithm for isomorphism testing, see \cite{GKMS} for an introduction to the subject. To show the relationship between NDF-equivalence and color refinement, we only need to recall its definition. At the start of this algorithm, a single color is associated to all nodes. In each round (iteration) of color refinement, two nodes $u$ and $v$ with the same color are assigned different colors only if there is a color, say $c$, such that the number of neighbors of $u$ with color $c$ is different than the number of neighbors of $v$ with color $c$. So, after the first round of color refinement the set of nodes of the graph are partitioned according to their degrees, that is $u$ and $v$ have the same color if and only if they have the same degree. After the second round, two nodes $u$ and $v$ have the same color if and only if for every occurring degree, say $k$, the numbers of neighbors of $u$ of degree $k$ is equal to the numbers of neighbors of $v$ of degree $k$. This is exactly the NDF-partition. 

We shall come back to this subject in Subsection \ref{subsec:rndfc-equivalence}, where we show how RNDFC-equivalence (a version of higher order NDF-equivalence) performs better than color refinement in certain examples. 
	

\section{NDF Embeddings for Dynamic Graphs}
\label{sec:dynamic}
	
To address the problems arising from the abundance of possible degrees in big graphs and the ever-changing nature of dynamic graphs appearing in real world applications, our main approach is to replace the degrees list appearing in the definition of the minimal NDF, see Definition \ref{def:ndf_minimal}, with a list of intervals and then compute the NDF vector representation of nodes with respect to the intervals in this list. To this work, we have to impose three necessary conditions:
\begin{itemize}
	\item[(1)] The union of intervals must cover all natural numbers. 
	\item[(2)] The intervals must be mutually disjoint, that is the intersection of every two different intervals must be empty. 
	\item[(3)] The list of starting numbers of intervals must be in ascending order. 
\end{itemize}

The first two conditions mean that the set of intervals must be a partition of the set of natural numbers. The third condition determines the order of components of the NDF vector. These conditions imply that the first interval must start with 1 and the last interval has to be infinite (with no end). Furthermore, since we intend to use the NDF vectors as features in machine learning algorithms, we impose a rather non-mandatory condition too:
\begin{itemize}
	\item[(4)] The last interval must include the maximum degree of nodes in the graph. 
\end{itemize}
Condition (4) guarantees that the last entry of the NDF vector of (at least) one node is non-zero. Otherwise, the last entry would be redundant. One may wish to impose a stronger version of Condition (4) to ensure that every interval in the intervals list contains some occurring degree, but this can be a requirement that is very hard to satisfy in general. We formalize requirements (1)-(4) in the following definition:
\begin{definition}
	\label{def:dynamic-ndf}
	Let $G=(V, E)$ be an undirected graph.
	\begin{itemize}
		\item [(i)] An {\bf intervals list} $\mathcal{I}$ for graph $G$ can be realized by a finite sequence $\{n_1, n_2, \cdots, n_m\}$ of natural numbers such that $1 = n_1< n_2 < \cdots < n_m \leq d$, where $d$ is the maximum degree of nodes in $G$. We call such a sequence the {\bf list of starting points} of the intervals list.  Then the intervals list $\mathcal{I}$ is the ordered set $\{I_1, I_2, \cdots, I_m\}$, where $I_i = \{n_i, n_i +1, \cdots, n_{i+1}-1\}$ for $i=1,\cdots,m - 1$ and $I_m = \{n_m, n_m +1, \cdots\}$. 
		\item [(ii)] To any given intervals list $\mathcal{I}$ for graph $G$ as above, we associate a vector representation $dndf: V \ra \mathbb{R}^m$ called the {\bf dynamic neighbors degree frequency embedding} of nodes of $G$ defined as $dndf(v) = (dndf(v)_1,\cdots, dndf(v)_m)$ where for $i=1,\cdots,m$, $dndf(v)_i$ is the number of neighbors of node $v$ whose degrees lie in $I_i$.
	\end{itemize}
\end{definition}

What we are referring to as an ``interval'' here can also be called a ``range''. However we prefer the word interval because each interval is indeed a usual interval in real numbers restricted to natural numbers. Although the definition of a dynamic NDF embedding depends on the intervals list $\mathcal{I}$ employed in its construction, for the sake of simplicity, we omit the intervals list $\mathcal{I}$ from our notation. When it is important one may want to emphasis this dependency by denoting the dynamic NDF by $dndf_{\mathcal{I}}$. 

It is clear from the above definition that the dimension of vectors obtained in dynamic NDF equals the number of intervals in the chosen intervals list. In the following example, we explain that our definition of dynamic NDFs has enough flexibility to cover dynamic versions of vanilla and minimal NDF as well as degree centrality. 

\begin{example} In this example, let $G$ be an undirected graph.
	\label{exam:dynamic-v-m-ndf}
	\begin{itemize}
		\item[(i)]  Let $d$ denote the maximum degree of nodes in $G$. Then the intervals list $\mathcal{I}$ constructed from the list of starting points $\{1, 2, \cdots, d\}$ gives rise to the {\it dynamic version of vanilla NDF}. For future reference, we call $\mathcal{I}$ the {\it vanilla intervals list}.     
		\item[(ii)] As Definition \ref{def:ndf_minimal}, let $\{d_1,\cdots, d_m\}$ be the ascending list of all occurring degrees of nodes in $G$. If $d_1\neq 1$, we change it to 1. Then it is a list of starting points and defines an intervals list, which we call it the {\it minimal intervals list}. The dynamic NDF associated to this intervals list is the {\it dynamic version of minimal NDF}. 
		\item[(iii)] The singleton $\{1\}$ can be considered as a list of starting points. The intervals list associated to it consists of only one interval which is the whole natural numbers and the dynamic NDF defined on top of this intervals list is simply the assignment of degree to each node of the graph. Up to some normalizing factor, this one dimensional vector representation of nodes is the same as degree centrality on $G$. 
	\end{itemize}
\end{example} 

The above example demonstrates the relationship between dynamic NDF embeddings and previously defined NDF embeddings. In the following definition, we present two general methods to construct intervals lists. Based on the intuition gained from network science, we assume that most nodes of a graph have low degrees and substantially fewer nodes have high degrees. This suggests that we devote more intervals to lower degrees and fewer intervals to higher degrees. On the other hand, nodes with higher degrees appear to be the neighbors of much greater number of nodes, and so it is equally important to differentiate between high degree nodes too. We also note that the distance between higher degrees are much bigger than the distance between lower degrees, which suggests smaller intervals around lower degrees and bigger intervals around higher degrees. Based on these observations, we propose an ``adjustable design'' for partitioning the set of natural numbers into intervals lists. Therefore, we include some parameters in our construction to create more flexibility and control on the design of intervals lists. 

\begin{definition} Let $G$ be an undirected graph and let $d$ be the maximum degree of nodes in $G$. We set a maximum length $m$ for intervals.
	\label{exam:ratio-dynamic-ndf}
	\begin{itemize}
		\item[(i)] As the first approach, we try to partition the set of natural numbers into intervals of fix length $m$ (besides the last interval which has to be infinite!). Due to the possible remainder, we let the potentially smaller interval to be the first interval. Therefore our algorithm for choosing a ``uniform'' list of starting points of the intervals list would be as follows:\\
		\begin{algorithm2e}[H]
			\KwIn{
				Integers $d \geq m \geq 1$ and $d>1$\\
				\tcp{$d$: the maximum degree of nodes in $G$}
				\tcp{$m$: the maximum length of intervals} }
			\KwOut{\text{A list of starting points}}
			$next\_point \leftarrow d$\\
			\tcp{Assign an empty list to $starting\_points$}
			$starting\_points \leftarrow [\ ]$ \\
			
			\While{$next\_point > 0$}{
				$starting\_points${\bf .append($next\_point$)}\\
				$next\_point \leftarrow next\_point - m$\\
				\If{$next\_point \leq 1$}{
					$starting\_points${\bf .append($1$)}\\
					{\bf break}\\
				}
			}
			\tcp{ Reverse the order of $starting\_points$}
			$starting\_points${\bf .reverse()}\\
			\KwRet{$starting\_points$}
			\caption{\label{algo:uniform_starting_points} Generating a uniform list of starting points}
		\end{algorithm2e}   
		\item[(ii)] As the second approach, we try to generate an intervals list such that the lengths of the intervals are increasing. In Algorithm \ref{algo:increasing_starting_points}, we set a starting length $s$ and a maximum length $m$ for intervals. Then in each step, we increase the length of intervals by the ratio $r$ until we reach the maximum length.\\
		\begin{algorithm2e}[H]
			\KwIn{ Integers $d \geq m \geq s \geq 1$ and a real number $r > 1$\\
				\tcp{$d$: the maximum degree of nodes in $G$}
				\tcp{$m$: the maximum length of intervals}
				\tcp{$s$: the starting length of intervals}
				\tcp{$r$: the ratio of the length of intervals} }
			\KwOut{\text{A list of starting points}}
			$next\_length \leftarrow s$\\
			$starting\_points \leftarrow [1, s+1]$ \\
			\While{$starting\_points[-1] < d$}{
				$next\_length \leftarrow$ {\bf min}$(next\_length * r, m)$\\
				$next\_point \leftarrow$ {\bf int}$(starting\_points[-1] + next\_length)$\\
				$starting\_points${\bf .append($next\_point$)}\\
			}
			\If{$starting\_points[-1] > d$}{
				$starting\_points${\bf .pop($-1$)}\\
			}
			\KwRet{$starting\_points$}
			\caption{\label{algo:increasing_starting_points} Generating an increasing list of starting points}
		\end{algorithm2e}
	\end{itemize}
\end{definition}

One notes that if we set $s=m$ (the starting length being equal to maximum length) in Algorithm \ref{algo:increasing_starting_points}, we do not necessarily get the same intervals list as obtained in Algorithm \ref{algo:uniform_starting_points} with the same maximum length, even though there is no increase in the length of intervals in both cases. In this scenario, the intervals list obtained from Algorithm \ref{algo:increasing_starting_points} consists of intervals of fixed length $m$ except the last interval that contains the maximum degree of nodes and has to be infinite. However the intervals list obtained from Algorithm \ref{algo:uniform_starting_points} consists of intervals of length $m$ and the last interval which starts from the maximum degree of nodes and possibly the first interval that might have length less than $m$. 

\begin{example}
	\label{exam:main_graph_dndf}
	Let $G$ be the graph shown in Figure \ref{fig:main-graph}. We define two intervals lists as follows:
	\begin{eqnarray*}
		\mathcal{I}_1 &=& \{\{1, 2\}, \{3, 4, \cdots\} \} \\
		\mathcal{I}_2 &=& \{\{1\}, \{2, 3\}, \{4, 5, \cdots\} \} 
	\end{eqnarray*} 
	The intervals list $\mathcal{I}_1$ can be generated by Algorithm \ref{algo:increasing_starting_points} by setting $s=2$, $m$ any integer greater than or equal to $3$ and $r$ any real number greater than or equal to $1.5$. Similarly, $\mathcal{I}_2$ can be generated by Algorithm \ref{algo:increasing_starting_points} by setting $s=1$, $m=2$ and $r$ any real number greater than or equal to $2$. The dynamic NDF embeddings of nodes of $G$ with respect to $\mathcal{I}_1$ and $\mathcal{I}_2$ are shown in Table \ref{table:g_dndf}. We notice that by reducing the dimension of NDF embedding from 5 (for vanilla NDF) to 3 (for dynamic NDF w.r.t $\mathcal{I}_2$) or even to 2 (for dynamic NDF w.r.t $\mathcal{I}_1$), we do not lose a whole lot of precision. One can expect this phenomenon to occur more often and more notably in graphs whose maximum degrees are very large.
	
	\begin{table}
		\begin{tabular}{|c|c||c|c|}\hline
			Node(s) & Dynamic NDF w.r.t $\mathcal{I}_1$ & Node(s) & Dynamic NDF w.r.t $\mathcal{I}_2$\\\hline
			A & (2, 3)& A & (1, 2, 2)  \\\hline
			D & (2, 2)& D & (2, 2, 0)  \\\hline
			K & (3, 1)& K & (1, 3, 0)   \\\hline
			Y, H & (1, 3)& Y, H & (0, 2, 2)   \\\hline
			B, I & (0, 3)& B & (0, 1, 2)   \\\hline
			& & I & (0, 0, 3)   \\\hline
			E & (1, 2)& E & (0, 2, 1)  \\\hline
			C, F, M, N & (1, 1)& C, M, N & (0, 1, 1)  \\\hline
			& & F & (0, 2, 0)   \\\hline
			L & (0, 2) & L & (0, 0, 2)   \\\hline
			J, P, Q, R & (0, 1) & J, P, Q, R & (0, 0, 1)   \\\hline
		\end{tabular}
		
		\caption{\label{table:g_dndf} Two dynamic NDF embeddings of nodes of graph $G$, shown in Figure \ref{fig:main-graph} with respect to the intervals lists  $\mathcal{I}_1$, $\mathcal{I}_2$ explained in Example \ref{exam:main_graph_dndf}. }
	\end{table}
\end{example}

The above algorithms demonstrate only two possible ways of generating starting points for intervals lists, and consequently two ways of defining reasonable dynamic NDFs on graphs.

\begin{remark}
	\label{rem:costume-sp}
	The only feature of the graph used in these methods is the maximum degree of nodes in the graph and the rest of parameters are set by the user.  The best way of defining an intervals list for a specific graph depends on the distribution of degrees of nodes in natural numbers. Therefore, for a given graph with a particular structure, in order to define the most relevant dynamic NDF embedding, one has to incorporate more detailed information about the set of occurring degrees. There are also situations that we need to define a common intervals list for two or more graphs. In these situations, we choose an auxiliary point less than the maximum degree (call it the ``last point''). Then, we use one of the above algorithms to choose a list of starting points up to the last point. To cover the rest of degrees, we choose a complementary list of points according to the task in hand and the set of degrees of nodes. The concatenation of these two lists of starting points yields the final list of starting points.  For example, assume we have a graph whose set of occurring degrees up to 100 can be perfectly covered by Algorithm \ref{algo:increasing_starting_points} and it has only 5 nodes whose degrees are greater than 100, say $\{108, 111, 121, 148, 293\}$. We let the last point be 100 and choose the primary list $SP_1$ of starting points using Algorithm \ref{algo:increasing_starting_points}. We heuristically set the complementary list $SP_2$ of points to be $\{103, 116, 136, 201\}$. Then the concatenation $SP_1 + SP_2$ of these two lists is our final list of starting points. While choosing the complementary list $SP_2$, we didn't care about any particular growth rate for the lengths of intervals. Instead, We tried to impose two things: First, to avoid any empty interval, an interval which contains none of the degrees in the set $\{108, 111, 121, 148, 293\}$. Secondly, intervals should contain some margin around the degrees they contain, (e.g. the first interval contains some numbers less than 108 and some numbers greater than 111). We use this method several times in this article to build mixed lists of starting points, for instance see Examples \ref{exam:pagerank-random}, \ref{exam:wiki-vote-pagerank} and \ref{exam:epinions-pagerank}.
\end{remark}

We note that by choosing appropriate values for the parameters of the above algorithms, we can control the dimension of the vector representation obtained from dynamic NDF embeddings. For instance, two and three dimensional dynamic NDF embeddings of the graph shown in Figure \ref{fig:main-graph} are discussed in Example \ref{exam:main_graph_dndf}. To some extent, this resolves the ``curse of dimensionality'' explained after Definition \ref{def:ndf_vanilla}. But aggressively reducing the dimension of the embedding by this method imposes a coarse aggregation among nodes, because all neighbors of a node whose degrees lie in an interval are considered the same type. An alternative way of reducing the dimension of the embedding would be to apply a deep neural network algorithm equipped with some appropriate objective function on the vectors obtained from dynamic NDFs. Then the number of neurons in the last layer (before applying the objective function) determines the dimension of the final embedding. Since the objective function should be designed according to the downstream tasks, we postpone the detailed analysis of these types of the compositions of neural networks and NDF embeddings to future works. 

Dynamic NDF embeddings employ some sort of aggregation while counting the degree frequency of neighbors. Therefore, some precision has been lost and we cannot repeat the content of Subsection \ref{subsec:tree-approx} for dynamic NDFs. However, it is still true that the degree of a node equals the sum of all entries of its dynamic NDF vector, i.e. $deg(v) = s_1(v)=\sum_{i = 1}^{m} dndf(v)_i$ for all nodes $v$ in the graph. 

The content of Subsection \ref{subsec:isom} can be repeated for dynamic embeddings as well, although it is obvious that dynamic NDF embeddings have less expressive power than vanilla and minimal NDF embeddings, see Example \ref{exam:main_graph_dndf}. To some extent, the expressive power of dynamic NDF embeddings can be increased by considering higher order NDF embeddings, which is the subject of the next section.

\quad

\section{Higher Order NDF Embeddings as Matrix Representations of Nodes}
\label{sec:ndf_mat}

So far for NDF embeddings, we have been considering the degree frequency of neighbors of a node and driving a vector representation from it. For higher order generalizations of NDF embeddings, we have two natural choices to drive vectors from circles around the node: 

In the first approach, we consider the previously defined NDF vector as the radius zero NDF. That is, for a node $v$, we computed neighbors degree frequency of elements of $C_0(v)=\{v\}$, the circle of radius zero centered at $v$. So naturally for $n > 0$, the radius $n$ NDF vector would be an aggregation of neighbors degree frequency of elements of $C_n(v)$. The aggregation method we adopt here is simply the mean of NDF vectors of elements of $C_n(v)$, but one can try other aggregation methods, depending on downstream tasks. More advanced option would be to learn the aggregation method, perhaps, similar to what has been done in \cite{HYL}, but in the level of each circle. Computing mean amounts to dividing the sum of NDF vectors of elements of $C_n(v)$ by the size of $C_n(v)$, i.e. $s_n(v)$. Dividing by $s_n(v)$ is necessary, because it makes the set of NDF vectors associated to different radiuses more homogeneous and comparable to each other. Whenever $C_n(v)$ is empty for some $n$, the radius $n$ NDF vector is simply defined to be the zero vector.

In the second approach, we consider the previously defined NDF vector as the radius one degree frequency (DF). The justification is that the NDF vector of a node $v$ is computed by counting the ``degree frequency'' of elements of $C_1(v)$. Therefore for an arbitrary natural number $n$, the order $n$ generalization of the NDF vector could be the degree frequency of elements of $C_n(v)$. This method produces a sequence of vectors with sharply varying magnitude and thus not comparable to each other, for instance see Examples \ref{exam:main_graph_matrix_reps}(iii) below. In order to normalize the degree frequency vectors in different radiuses, we divide each nonzero DF vector of radius $n$ by $s_n(v)$. Again when $C_n(v)$ is empty (i.e. $s_n(v)=0$), the DF vector of radius $n$ is defined to be the zero vector. One checks that $s_n(v)$ is exactly the $\ell^1$-norm of the DF vector of radius $n$ centered at node $v$, so this method of normalization is called the {\bf $\ell^1$-normalization}. It converts every nonzero DF vector (including the previously defined NDF vector) into a probability vector. In this approach, one should also note that we use the phrase ``degree frequency'' (DF) in lieu of ``neighbors degree frequency'' (NDF), because the degree frequency of very elements of circles are computed, not the ones of their neighbors! 

Finally, these methods of extending NDF vector representations even without dividing by the size of circles (normalization) can be useful in certain topics, so we record both versions for later references. The formal definitions of these generalizations of NDF embeddings are as follows:

\begin{definition} Let $G=(V,E)$ be an undirected graph and let $\mathcal{I}=\{I_1,I_2,\cdots,I_m\}$ be an intervals list for $G$.
	\label{def:matrix-reps}
	\begin{itemize}
		\item [(i)] The {\bf order $r$ matrix of neighbors degree frequency of circles} ({\bf NDFC} matrix) with respect to the intervals list $\mathcal{I}$ associates an $(r+1) \times m$ matrix $ndfc_r(v)$ to each node $v\in V$ such that the $k$-th row  of $ndfc_r(v)$ is computed as follows:
		\begin{equation*}
			\frac{1}{s_{k-1}(v)} \sum_{u\in C_{k-1}(v)} dndf(u), \qquad \text{when } s_{k-1}(v)\neq 0,
		\end{equation*}
		and the zero vector of dimension $m$ when $s_{k-1}(v)=0$. 
		
		The order $r$ {\bf RNDFC} matrix, the {\bf raw} version of NDFC matrix, of a node $v$ is denoted by $rndfc_r(v)$ and is defined similarly except that the $k$-th row of $rndfc_r(v)$ is not divided by $s_{k-1}(v)$.
		\item [(ii)] The {\bf order $r$ matrix of circles degree frequency} ({\bf CDF} matrix) with respect to the intervals list $\mathcal{I}$ assigns an $r\times m$ matrix $cdf_r(v)$ to each node $v\in V$ such that when $s_i(v)\neq 0$, the $(i, j)$-th entry of $cdf_r(v)$ is the number of nodes in $C_i(v)$ whose degrees belong to $I_j$ divided by $s_i(v)$, that is
		\begin{equation*}
			cdf_r(v)_{ij} = \frac{|\{u\in C_i(v); deg(u)\in I_j\}|}{s_i(v)},\qquad \text{when } s_i(v)\neq 0.
		\end{equation*}
		When $s_i(v)=0$, the entire $i$-th row of $cdf_r(v)$ is defined to be the zero vector of dimension $m$. 
		
		The order $r$ {\bf RCDF} matrix, the {\bf raw} version of the CDF matrix, of a node $v$ is denoted by $rcdf_r(v)$ and is defined similarly except that the $(i, j)$-th entry of $rcdf_r(v)$ is not divided by $s_i(v)$.  
	\end{itemize}
\end{definition}

One notes that the first row of NDFC (respectively RNDFC and RCDF) matrix is the same as the dynamic NDF vector. However the first row of CDF matrix is the $\ell^1$-normalized version of the dynamic NDF vector, i.e. $dndf(v) / \|dndf(v)\|_1$ where $\|-\|_1$ is the $\ell^1$-norm. In the following examples, we highlight some of the features of these matrix representations of nodes of a graph:

\begin{examples} 
	\label{exam:main_graph_matrix_reps}
	Let $G$ be the graph in Figure \ref{fig:main-graph}. Let $\mathcal{I}_0$ be the vanilla intervals list of $G$ and let $\mathcal{I}_1$ be as defined in Example \ref{exam:main_graph_dndf}. In the following, we say a matrix representation separates two nodes if its values on these nodes are different. We round all float numbers up to three decimal digits.
	\begin{itemize}
		\item [(i)] The order 3 NDFC matrix w.r.t. $\mathcal{I}_0$ for nodes $H$ and $Y$ are as follows:
		\[
			ndfc_3(H) =  
			\left[
			\begin{array}{lllll}
				0 . & 1 . & 1 . & 1 . & 1 . \\
				0.25 & 0.75 & 0.5 & 1.75 & 0.25 \\
				0.167 & 0.667 & 0.333 & 0.667 & 0.5 \\
				0.5 & 0.5 & 1 . & 0.5 & 0 .
			\end{array}
			\right], \quad
			ndfc_3(Y) = 
			\left[
			\begin{array}{lllll}
				0 . & 1 . & 1 . & 1 . & 1 . \\
				0.25 & 0.5 & 0.5 & 2 . & 0.25 \\
				0.2 & 0.8 & 0.4 & 0.4 & 0.6 \\
				0.4 & 0.6 & 0.8 & 0.6 & 0 . 
			\end{array}
			\right]
		\]
		Although the vanilla NDF vectors of $H$ and $Y$ are the same, others rows of the NDFC matrix of these nodes are different. In fact, even the order 1 NDFC matrix representation of nodes of $G$ w.r.t. $\mathcal{I}_0$ can separate all nodes of $G$, except $R$ and $Q$ for which there is an automorphism of $G$ that flips them, so they must have identical neighborhood statistics. This emphasizes the fact that higher order NDF representations of nodes have more expressive power, namely they are better at revealing the differences between neighborhoods of the nodes of a graph. This is in line with our discussion in Subsection \ref{subsec:isom} and will be explained more in Subsection \ref{subsec:rndfc-equivalence}.
		\item [(ii)] The order 5 NDFC matrix w.r.t. $\mathcal{I}_1$ for nodes $B$ and $I$ are as follows::
		\[
		ndfc_5(B) =  
		\left[
		\begin{array}{ll}
			0 . & 3 . \\
			1.667 & 2.333 \\
			0.571 & 1.571 \\
			0.333 & 2 . \\
			2 . & 1 . \\
			0 . & 1 .
		\end{array}
		\right], \quad
		ndfc_5(I) = 
		\left[
		\begin{array}{ll}
			0 . & 3 . \\
			1.667 & 2.333 \\
			0.8 & 1.6 \\
			0.333 & 1.667 \\
			1.333 & 1.667 \\
			0 . & 1 .
		\end{array}
		\right]
		\]
		Although nodes $B$ are $I$ have different vanilla NDF vectors, the first two rows of their NDFC matrices w.r.t. $\mathcal{I}_1$ are the same. In fact, it take the order 2 NDFC matrix representation of nodes of $G$ w.r.t. $\mathcal{I}_1$ to separate all nodes of $G$, (of course, except $R$ and $Q$).
		\item [(iii)] The order 3 CDF and RCDF matrices w.r.t. $\mathcal{I}_1$ of nodes $C, F, M$ and $N$ are as follows:
		\[
		cdf_3(C) =  
		\left[
		\begin{array}{ll}
			0.5 & 0.5 \\
			0.2 & 0.8 \\
			0.5 & 0.5 
		\end{array}
		\right], 
		cdf_3(F) =  
		\left[
		\begin{array}{ll}
			0.5 & 0.5 \\
			0. & 1. \\
			0.6 & 0.4
		\end{array}
		\right], 
		cdf_3(M) =  
		\left[
		\begin{array}{ll}
			0.5 & 0.5 \\
			0. & 1. \\
			0.8 & 0.2 
		\end{array}
		\right], 
		cdf_3(N) =  
		\left[
		\begin{array}{ll}
			0.5 & 0.5 \\
			0.5 & 0.5 \\
			0. & 1.
		\end{array}
		\right]
		\]
		\[
		rcdf_3(C) =  
		\left[
		\begin{array}{ll}
			1 & 1 \\
			1 & 4  \\
			2 & 2 
		\end{array}
		\right], \quad
		rcdf_3(F) =  
		\left[
		\begin{array}{ll}
			1 & 1 \\
			0 & 3 \\
			3 & 2 
		\end{array}
		\right], \quad
		rcdf_3(M) =  
		\left[
		\begin{array}{ll}
			1 & 1 \\
			0 & 4 \\
			4 & 1 
		\end{array}
		\right], \quad
		rcdf_3(N) =  
		\left[
		\begin{array}{ll}
			1 & 1 \\
			2 & 2 \\
			0 & 2 
		\end{array}
		\right]
		\]
		One can easily check that the CDF matrix representation has to be of order at least 3 to be able to separate the nodes of $G$ (except $R$ and $Q$), while order 2 is enough for RCDF to do the same task.  This example clearly shows that RCDF can be more efficient than CDF in separating the nodes of $G$.
	\end{itemize}
\end{examples}

To use our matrix representations of nodes in our exemplary machine learning algorithms in next sections, we shall apply certain aggregations to turn them into 1-dimensional arrays. However the present 2-dimensional array structure of these representations offers certain secondary applications.


\subsection{Shapes of neighborhoods of nodes}
\label{subsec:shape}
\qquad\\

The matrix representations of nodes as 2-dimensional arrays can be used to produce a simple visualization of the neighborhoods of nodes in a graph. For this purpose, we use only NDFC matrices because of several reasons: (1) In NDFC matrices, all rows are NDF vectors or their means (aggregations). So, there is a meaningful relationship between different rows. (2) The rows of RNDFC and RCDF matrices are not comparable to each other and heavily populated rows suppress less populated rows. So one cannot compare them by looking at the visualization. For example, consider $rcdf_3(C)$ in Part (iii) of Example \ref{exam:main_graph_matrix_reps}, there is no obvious and meaningful way to compare the second row with the first row. (3) The row-by-row $\ell^1$-normalization of CDF matrices might suppress some important values. For instance, consider $cdf_3(F)$ in Part (iii) of Example \ref{exam:main_graph_matrix_reps}. In the matrix $rcdf_3(F)$, we observe that there are two nodes in $C_3(F)$ whose degrees are greater than 2, while only one node in $C_1(F)$ has a degree greater than 2. However, in $cdf_3(F)$, the entry $(3,2)$ is less than the entry $(1,2)$, which is clearly counter-intuitive.

We slightly modify NDFC matrices to obtain a more extensive visualization. With the notation of Definition \ref{def:matrix-reps}, for every node $v$ in the graph, we define a vector of dimension $m$ whose all entries are zero except the $j$-th entry which is 1 and $j$ is the index such that the degree of $v$ belongs to $I_j$. We call this vector the {\bf degree vector} of $v$ with respect to the intervals list $\mathcal{I}$. We add this vector as the first row to the NDFC matrix and shift other rows below. In this way, we can view the degree of the node as the NDF vector of an imaginary circle of radius -1 around each node and take it into account while visualizing other NDF vectors of circles. Let us call the matrix representation obtained in this way the {\bf visual NDFC (or VNDFC) matrix representation} of nodes of the graph and denote the order $r$ VNDFC matrix representation of a node $v\in V$ by $vndfc_r(v)$.

\begin{example}
	\label{exam:main-visual}
	We choose 7 nodes $A$, $J$, $C$, $D$, $Q$, $M$, and $N$ of the graph $G$ in Figure \ref{fig:main-graph} and use the Python package matplotlib to draw the (gray) color map of the order 7 VNDFC matrix representations of these nodes w.r.t. the vanilla intervals list. The order 7 is chosen because the diameter of $G$ is 7 and all 9 rows of the $vndfc_7(Q)$ are non-zero. Figure \ref{fig:nodes-viz} illustrates the visualizations of these nodes. It is worth noting that since nodes $J$ and $Q$ are connected to the rest of graph via nodes $A$ and $D$, respectively, they have very similar visualizations to the ones of $A$ and $D$, respectively. One observes less similarity between the visualizations of nodes $A$ and $C$, because $A$ only connects $C$ to the left part of the graph. One also notices some similarity between visualizations of nodes $M$ and $N$, which is not surprising regarding the similarity of their neighborhoods in $G$.
	\begin{figure}
		\centering
		\includegraphics[width=450px]{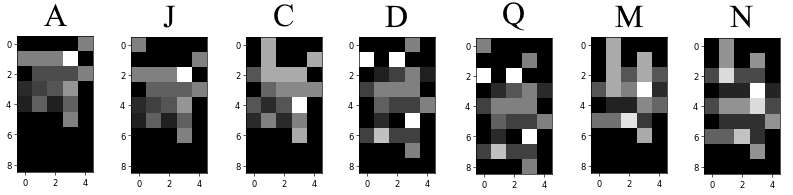}
		\caption{Visualizations of several selected nodes of the graph depicted in Figure \ref{fig:main-graph} as described in Example \ref{exam:main-visual}.} \label{fig:nodes-viz}
	\end{figure}
\end{example}

The representation of nodes by 2-dimensional arrays and the above visualizations suggest the idea that one can feed the higher order NDF matrices into convolutional neural network models and extract some patterns. However, unfortunately, there is an important reason that prevents two dimensional convolutions to be able to detect and reveal meaningful patterns in these matrices. The reason is that the rows and columns of these matrices have totally different natures, and so unlike images, the patterns in higher order NDF matrices are not equivariant under most two dimensional geometric transformations. That being said, inspired by the observations made in Example \ref{exam:main-visual}, there are some hopes for one dimensional convolutional operations to be useful in detecting certain patterns among these matrices. Specifically, we expect the vertical (column-wise) convolution operations be able to detect hidden patterns in higher order NDF matrices. Further evidences for this idea are provided by the success of applying certain column-wise aggregations to the RCDF matrix representations of nodes in Section \ref{sec:aggregations-closeness}. Example \ref{exam:fb-co-pagerank-conv} shows an instance of such a convolutional neural networks, though it has less success than the other neural network that consists of fully connected linear layers. 


\subsection{RNDFC-equivalence vs color refinement}
\label{subsec:rndfc-equivalence}
\qquad\\

We have four options for higher order generalizations of NDF-equivalence; equivalences based on NDFC, RNDFC, CDF, and RCDF matrix representations of nodes. During normalization of rows in NDFC and CDF matrices, we loose some information, so raw versions of these matrix representations, i.e. RNDFC and RCDF, work better for separating the nodes of a graph. On the other hand, since RNDFC matrices are computed based on neighbors degree frequency, they are more similar to color refinement and heuristically better than RCDF matrices, which are computed based on degree frequency. Therefore we work only with vanilla RNDFC-equivalence of graphs and nodes here. The definition of {\bf RNDFC-equivalence} is the same as Definition \ref{def:ndf-equivalence}. Also, one checks that isomorphism of graphs implies RNDFC-equivalence. 

Here, we content ourselves with two examples in which RNDFC-equivalence can detect two non-isomorphic graphs while color refinement cannot and postpone the detailed study of RNDFC-equivalence as a tool for graph isomorphism testing to our future works.

\begin{example}
	\label{exam:2-regulars}
	Let $G_1$ be the cycle graph of length 6 and let $G_2$ be the disjoint union of two cycle graphs of length 3. They are both 2-regular graphs, and so color refinement does not change the colors of nodes. Therefore it cannot tell the difference between these graphs. On the other hand the diameter of $G_2$ is 1. This implies that only first two rows of RNDFC-matrices of its nodes are non-zero. However, the diameter of $G_1$ is 3 and so there are 4 non-zero rows in the order $n$ RNDFC matrices of nodes of $G_1$ for all $n\geq 3$. This clearly shows that these graph are not RNDFC-equivalent and so they are not isomorphic. 
\end{example} 

In the above example, we didn't have to even compute RNDFC matrices and our argument was based on the diameter of the graphs. This shows the advantage of RNCDF-equivalence versus color refinement. But the main superiority of RNCDF-equivalence over color refinement is based on the fact that it labels nodes by means of concrete matrices. Therefore in cases that the color refinement and RNCDF-equivalence both partition the set of nodes similarly, the matrices associated to each equivalence class might tell the difference between two graphs. 

\begin{figure}
	\centering
	\includegraphics[width=330px]{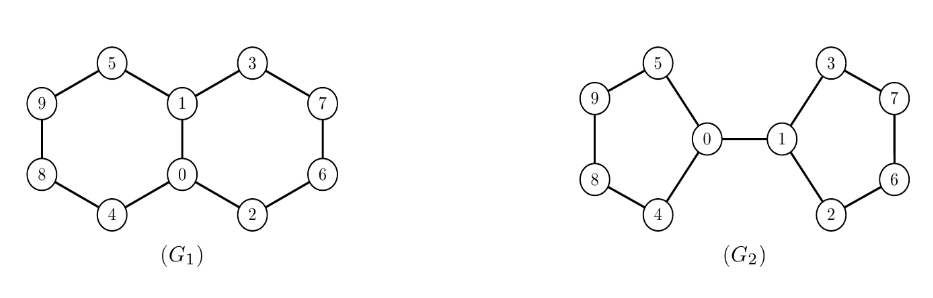}
	\caption{The graphs in Example \ref{exam:8-shape} } \label{fig:8-shape}
\end{figure}

\begin{example}
	\label{exam:8-shape}
	Let $G_1$ and $G_2$ be the graphs illustrated in Figure \ref{fig:8-shape}. Because of vertical and horizontal symmetries in both graphs, the color refinement partitions them similarly as follows:
	\[
	P_1 = \{0,1\}, \quad P_2 = \{2,3,4,5\}, \quad P_3 = \{6,7,8,9\}.
	\]
	So color refinement cannot tell they are not isomorphic. Again because of those symmetries, the RNDFC-partitions of these graphs are the same too. However, the RNDFC matrices associated to classes $P_2$ and $P_3$ are different in each graph. To see this, we list the order 5 RNDFC matrices (w.r.t. vanilla intervals list) of each RNDFC-equivalence class in graphs $G_1$ and $G_2$:
	\begin{itemize} 
		\item [($G_1$)] 
		\[
		rndfc_5(P_1) =  
		\left[
		\begin{array}{lll}
			0 & 2 & 1 \\
			0 & 4 & 3 \\
			0 & 6 & 2 \\
			0 & 4 & 0 \\
			0 & 0 & 0 \\
			0 & 0 & 0  
		\end{array}
		\right], \quad
		rndfc_5(P_2) =  
		\left[
		\begin{array}{lll}
			0 & 1 & 1 \\
			0 & 4 & 1 \\
			0 & 5 & 2 \\
			0 & 4 & 2 \\
			0 & 2 & 0 \\
			0 & 0 & 0 
		\end{array}
		\right], \quad
		rndfc_5(P_3) =  
		\left[
		\begin{array}{lll}
			0 & 2 & 0 \\
			0 & 3 & 1 \\
			0 & 3 & 2 \\
			0 & 3 & 2 \\
			0 & 3 & 1 \\
			0 & 2 & 0 
		\end{array}
		\right].
		\]
		\item [($G_2$)] 
		\[
		rndfc_5(P_1) =  
		\left[
		\begin{array}{lll}
			0 & 2 & 1 \\
			0 & 4 & 3 \\
			0 & 6 & 2 \\
			0 & 4 & 0 \\
			0 & 0 & 0 \\
			0 & 0 & 0  
		\end{array}
		\right], \quad
		rndfc_5(P_2) =  
		\left[
		\begin{array}{lll}
			0 & 1 & 1 \\
			0 & 4 & 1 \\
			0 & 5 & 2 \\
			0 & 2 & 2 \\
			0 & 4 & 0 \\
			0 & 0 & 0 
		\end{array}
		\right], \quad
		rndfc_5(P_3) =  
		\left[
		\begin{array}{lll}
			0 & 2 & 0 \\
			0 & 3 & 1 \\
			0 & 3 & 2 \\
			0 & 2 & 1 \\
			0 & 2 & 2 \\
			0 & 4 & 0 
		\end{array}
		\right].
		\] 
		
	\end{itemize}
	Therefore, $G_1$ and $G_2$ are not RNDFC-equivalent. Consequently they cannot be isomorphic. In fact, order 3 RNDFC matrices were enough to prove this. 
\end{example}

\section{Learning PageRank}
\label{sec:pagerank}

In this section, we present several simple examples to show how the higher order NDF matrix representations can be fed as feature vectors into simple feedforward neural networks consisting of only fully connected linear layers to learn PageRanks of graphs. We run these experiments/examples on two undirected graphs appearing in social networks: The first graph is the network of facebook pages of 14113 companies. The second graph is the network of 27917 facebook pages of media. The datasets of these graphs are available in \cite{RA}. They are big enough to contain a sufficient number of nodes for training deep learning models and small enough to make it possible to test each sensible scenario in less than an hour on an ordinary PC. Along the way, we also examine various options for feature engineering and the architecture of the models. After training our model on the first graph, we modify the graph slightly and observe that the trained model still has a good predictive power on the modified version of the graph. This proves that our method is capable of handling dynamic graphs. In order to show that our method of learning PageRank is inductive, we next consider two random graphs generated by the Barabási–Albert algorithm with different seeds. We train our model on one of the graphs and show it is able to predict the PageRank of the other graph as well. 

\begin{example}
	\label{exam:fb-co-pagerank}
	Let $G=(V, E)$ be the graph of facebook pages of companies. The nodes represent the pages and edges represent the mutual likes. It has 14113 nodes and average degrees of nodes is 7. One finds more specifics in the webpage of the graph \cite{RA}. To produce an intervals list, we use algorithm \ref{algo:increasing_starting_points} with the parameters $s=1$, $m=50$, $d=\max_{v\in V}deg(v)$ and $r=1.3$. For the higher order NDF matrix representation, we choose the order 5 NDFC matrices associated to this intervals list. The result is 14113 matrices of size $6\times 17$ indexed by nodes of $G$. We flatten these matrices as vectors with 102 entries and the result is the set of feature vectors for training/testing of our model. For the target variable, we compute the PageRanks of nodes using the Python package NetworkX, \cite{HSS}. In order to establish a rough balance between target variables and feature vectors, we multiply all PageRanks of nodes by 10000. The dataset is shuffled and divided into train/test sets, 10000 nodes for training and the rest of 4113 nodes for evaluating the model success. The following feedforward neural network model is used to learn the PageRank of the graph:
	\begin{lstlisting}[caption={\label{listing:pagerank-fb-co} The feedforward neural network model, written in PyTorch, used for learning PageRank of the network of facebook pages of companies.}]
		class FFNN_model(nn.Module):
		def __init__(self):
		super().__init__()
		num_features = features.shape[1]  # in this example; 102
		self.fc1 = nn.Linear(num_features, 400)
		self.fc2 = nn.Linear(400, 800)
		self.fc3 = nn.Linear(800, 200)
		self.fc4 = nn.Linear(200, 64)
		self.fc5 = nn.Linear(64, 8)
		self.fc6 = nn.Linear(8, 1)
		self.dropout1 = nn.Dropout(0.4)
		self.dropout2 = nn.Dropout(0.3)
		self.dropout3 = nn.Dropout(0.5)
		
		def forward(self, X):
		X = torch.tanh(self.fc1(X))
		X = torch.relu(self.fc2(X))
		X = self.dropout1(X)
		X = torch.relu(self.fc3(X))
		X = self.dropout3(X)
		X = torch.relu(self.fc4(X))
		X = self.dropout2(X)
		X = torch.tanh(self.fc5(X))
		return self.fc6(X)
	\end{lstlisting}
	Our objective function is the Mean Square Error function and we apply Adam optimizer with the learning rate 0.001. In each epoch, the training dataset is shuffled again and divided into 25 batches of size 400, so 25 rounds of gradient decent is performed in each epoch. After training this model for 2000 epochs, our model is able to predict the PageRanks of nodes in the test set with around 90\% accuracy (more precisely; 9.651\% inaccuracy on average). One should note that due to some random factors in our algorithm, like shuffling the data and parameter initialization, retraining the model may result slightly different accuracy rate. After saving this model, we randomly remove 500 edges from $G$ and add 500 new edges in such a way the modified graph stays connected. We apply the same feature engineering on the modified graph and compute the PageRanks of its nodes using NetworkX as targets. Our saved model is able to predict the PageRanks of the nodes of the modified graph with an average inaccuracy rate of 9.783\%. 
\end{example}

We have also run an experiment to train a neural network model equipped with convolutional layers to learn the PageRanks of nodes. Although the final accuracy rate is worse than the above example, we present this experiment here as an example for typical convolution neural networks applied on higher NDF matrices, as explained at the end of Subsection \ref{subsec:shape}. 

\begin{example}
	\label{exam:fb-co-pagerank-conv}
	Every little detail of this experiment is similar to Example \ref{exam:fb-co-pagerank}, except that we do not reshape the matrices as vectors and the model starts with two convolution layers and a pooling is applied after the first convolution layer. Afterwards we flatten the data to continue with the fully connected linear layers.  The size of the kernels of convolution layers are $3\times 1$ and $4\times 1$, respectively. The specifics of these layers are listed in Listing \ref{listing:pagerank-fb-co-conv}.	
	\begin{lstlisting}[caption={\label{listing:pagerank-fb-co-conv} The definition of convolution layers of the convolutional neural network model, written in PyTorch, used for learning PageRank of the network of facebook pages of companies.}]
		# Convolutional layers
		final_out_channels = 6
		self.conv_layer_1 = nn.Conv2d(in_channels=1, out_channels=3, kernel_size=(3,1), stride=(1, 1), padding=(1,0),  padding_mode='zeros')
		self.conv_layer_2 = nn.Conv2d(in_channels=3, out_channels=final_out_channels, kernel_size=(4,1), stride=(1, 1), padding=(1,0),  padding_mode='zeros')
		self.pooling_1 = nn.MaxPool2d(kernel_size=(3,1), stride=None, padding=0, dilation=1, return_indices=False, ceil_mode=False)
	\end{lstlisting}
	
	The trained model predicts the PageRank of nodes in the test set with around 64\% accuracy on average. 
\end{example}

In the following example, we discuss our experiments concerning deep learning models for PageRank of the network of facebook pages of media:

\begin{example}
	\label{exam:fb-media-pagerank-ndfc}
	Let $G=(V, E)$ be the graph of facebook pages of media. The nodes represent the pages and edges represent the mutual likes. It has 27917 nodes and average degrees of nodes is 14. One finds more specifics in the webpage of the graph \cite{RA}. To build an intervals list, we apply algorithm \ref{algo:increasing_starting_points} with the parameters $s=1$, $m=70$, $d=\max_{v\in V}deg(v)$ and $r=1.3$. For the higher order NDF, we take the order 5 NDFC matrices w.r.t. to this intervals list. The result is matrices of size $6\times 23$, and after flattening them the feature vectors have 138 entries. Again the PageRanks of nodes are the target variables. To balance the target variables with the feature vectors, this time, we multiply PageRanks by 24000. The dataset is shuffled and divided into 20000 nodes for training and 7917 nodes for evaluating the model. The feedforward neural network model is the same as the model described in Example \ref{exam:fb-co-pagerank}. Again we train the model for 2000 epochs and the average inaccuracy of the prediction on the test set is 12.438\%. Obviously, there are some random processes in our training algorithm such as shuffling and splitting the data, initialization of parameters, etc. To make sure that these elements of randomness do not undermine the promise of our methods, we repeat the training, and in the second try, the average inaccuracy on the test set is 13.551\%. We have also trained the same model with the order 6 CDF matrix representation in lieu of the order 5 NDFC matrix representation and the trained model had the average inaccuracy of 28.44\% on test set. This suggests that the NDFC matrix representation is a better choice than the CDF matrix representation for learning PageRank. 
\end{example}

In order to demonstrate that our method of learning PageRank based on the higher order NDF matrix representation is inductive, we run another experiment in the following example.

\begin{example}
	\label{exam:pagerank-random}
	We use Python package NetworkX to build two random graphs with 20000 nodes. The method we apply to generate these random graphs is based on the preferential attachment and is called the dual Barabási–Albert algorithm, see \cite{M, HSS} for details. It has the following parameters:
	\begin{itemize}
		\item $n$: The number of nodes. We set $n=20000$ for both graphs.
		\item $p$: A probability, i.e. $0<p<1$. We set $p=0.5$ for both graphs.
		\item $m_1$: The number of edges to link each new node to existing nodes with probability $p$. We set $m_1=3$ for both graphs.
		\item $m_2$: The number of edges to link each new node to existing nodes with probability $1-p$. We set $m_2=1$ for both graphs.
		\item $seed$: The seed indicating the random number generation state. We set $seed=1$ and $seed=2$ for the first graph and the second graph, respectively.
		\item $initial\_graph$: The initial graph to start the preferential attachment process. This parameter is set ``None'' for both graphs.
	\end{itemize}
	The graphs we obtain both have around 40000 edges with an average degree of around 4. However the list of degrees occurring in these graphs are very different especially for higher degrees, For instance the maximum degree of the first graph is 280, while it is 314 for the second graph. Therefore to choose a common list of starting points, we have to use the method explained in Remark \ref{rem:costume-sp}, and our choice for the list of starting points is as follows:
	\[
	\{1, 2, 3, 4, 5, 7, 9, 11, 14, 18, 23, 29, 36, 44, 60, 80, 100, 127, 150, 165, 205\}.
	\]
	We compute the order 5 NDFC matrix representations of both graphs w.r.t. this list of starting points and obtain a matrix of size $6 \times 21$ for each node of these graphs. We reshape these matrices as 1-dimensional arrays with 126 entries and consider them as the feature vectors. As for targets, we compute the PageRanks of nodes of these graphs. To balance the target variables with the feature vectors, we multiply PageRanks by 1000. Afterwards, we train the same feedforward neural network model described in Example \ref{exam:fb-co-pagerank} with 10000 randomly chosen feature vectors of the first graph as the feature set and the PageRank of the corresponding nodes as the target set. The average error of the trained model on the rest of 10000 nodes is 8.296\%. Then, we apply the trained model to the whole dataset of the second graph. The average error of predicting PageRanks of nodes of the second graph is 8.926\%. 
\end{example}

This example shows that the ``local'' feature vectors obtained from the order 5 NDFC matrix representation of the first graph contain enough statistical information about the ``global'' structure of the graph to be useful for learning and predicting not only the PageRanks of nodes in the first graph, but also the PageRanks of nodes in the second graph. In other words, the above example confirms that our method is a local-to-global strategy. It also clearly shows that our method can learn PageRank using the data of a known graph, and then predict the PageRank of another graph, as long as these graphs share some statistical similarities. Hence our method of learning PageRank is inductive. We shall return to the subject of leaning of PageRank in Section \ref{sec:digraphs}, after adopting a modification inspired by directed graphs.


\section{Aggregations and learning the closeness centrality}
\label{sec:aggregations-closeness}

The operation of concatenating rows (resp. columns) is one of the most basic methods of aggregation on entries of a matrix which amounts to reshaping the matrix (resp. the transpose of the matrix) into a 1-dimensional array. Some of the other basic aggregation methods are minimum, maximum, summation, mean and median which can be performed column-wise, row-wise or on all entries of the matrix. However since the order of rows in our matrix representations comes from the Breadth-First Search (or equivalently, the radius of the circles around nodes), we have a natural order in each column. On the other hand, the order of columns comes from the natural increasing order of starting points of the underlying intervals list, so there is a meaningful order between entries of each row too. Exploiting these vertical and horizontal orderings gives us a pivot and more creative opportunities for feature engineering based on these matrix representations. The following definition which was inspired by our previous constructions in Section \ref{sec:parametric-centrality} is only an instance of new types of possible aggregations:    

\begin{definition}
	\label{def:parametric-agg}
	Let $G=(V, E)$ be a graph and let $\pi:V\ra M_{r\times m}$ be a matrix representation of nodes of $G$ as $r\times m$ matrices. For any given sequence of parameters $\Lambda=\{\lambda_i\}_{i=1}^{r}$ of real numbers, we define a vector representation $\pi_\Lambda:V\rightarrow \mathbb{R}^m$ as follows:
	\[
	\pi_\Lambda(v)_j = \sum_{i=1}^r \lambda_i \pi(v)_{i,j} \qquad \forall v\in V, \quad j=1,\dots,m.
	\]
	In other words, $\pi_\Lambda = M_\Lambda \pi$, where $M_\Lambda$ is the vector (or the row matrix) $\left[\lambda_1,\dots,\lambda_r\right]$. We call $\pi_\Lambda$ the {\bf parametric aggregation} of $\pi$ with respect to  $\Lambda$.
\end{definition}

More precisely, the above aggregation is a column-wise parametric aggregation. Obviously, we can repeat this definition for row-wise parametric aggregation which replaces each row by the weighted sum of its elements. Even more generally, one can consider an $r\times m$ matrix of parameters and perform a 2-dimensional aggregation of the given matrix representation or apply any other pooling technique to obtain a lower dimensional (matrix or vector) representation of nodes. Of course, they all depend on the prior knowledge about the structure of the graph and the downstream task in hand. Inspired by $p$-centrality function discussed in Section \ref{sec:parametric-centrality}, in the following definition, we introduce a special case of parametric aggregation for our matrix representations of nodes:

\begin{definition}
	\label{def:p-agg}
	Let $G=(V, E)$ and $\pi:V\ra M_{r\times m}$ be as Definition \ref{def:parametric-agg}. For a given real number $0<p<1$, the {\bf $p$-aggregation} of the matrix representation $\pi$ is the parametric aggregation of $\pi$ with respect to the parameter sequence $\Lambda=\{p^{i-1}\}_{i=1}^{r}$.  and it is denoted simply by $\pi_p$.
\end{definition}

Given a node $v$ in a graph, the sum of all entries of the $i$-th row of the RCDF matrix representation of $v$ is $s_i(v)$. So one can consider the rows of the RCDF matrix of $v$ as a decomposition of first few terms of the sequence $\{s_i(v)\}$ with respect to the underlying intervals list. Therefore the $p$-aggregation of the RCDF matrix is a decomposition of $p$-centrality function. Combining this observation with the intuition gained in Examples \ref{exam:main_parametric} and \ref{exam:three_classic} makes us to believe that the $p$-aggregation of the RCDF matrix contains a good deal of information about the closeness centrality of the graph. In the following examples, we feed the vectors obtained from the $p$-aggregation of the RCDF matrix representations into simple feedforward deep learning models to learn and predict closeness centrality of the graphs discussed in the previous section. 

\begin{example}
	\label{exam:fb-co-closeness}
	Let $G=(V,E)$ be the graph of facebook pages of companies discussed in Example \ref{exam:fb-co-pagerank}. 
	\begin{itemize}
		\item [(i)] Let $\mathcal{I}$ be the intervals list associated with the starting points generated by Algorithm \ref{algo:increasing_starting_points} according to the parameters: $s=1$, $m=35$, $d=\max_{v\in V}deg(v)$ and $r=1.5$, We set $p = 0.3$ and consider the $p$-aggregation of the order 4 RCDF matrix representation of nodes of $G$ w.r.t. the intervals list $\mathcal{I}$ as our feature set. We compute the closeness centrality of the nodes of the graph using NetworkX library, \cite{HSS}, and consider it as the target set. We employ a not-so-deep feedforward neural network model as follows:
		\begin{lstlisting}[caption={\label{listing:closeness-fb-co} The feedforward neural network model, written in PyTorch, used for learning closeness centrality of the network of facebook pages of companies.}]
			class FFNN_model(nn.Module):
			def __init__(self):
			super().__init__()
			num_features = features.shape[1]
			self.fc1 = nn.Linear(num_features, 64)
			self.fc2 = nn.Linear(64, 8)
			self.fc3 = nn.Linear(8, 1)
			self.dropout1 = nn.Dropout(0.3)
			
			def forward(self, X):
			X = torch.tanh(self.fc1(X))
			X = self.dropout1(X)
			X = torch.relu(self.fc2(X))
			return self.fc3(X)
		\end{lstlisting}
		We use the Mean Square Error function as the objective function and apply Adam optimizer. The nodes are shuffled and divided into train/test sets with 10000 nodes for training and 4113 nodes for testing. During each epoch, the training dataset is shuffled and divided into 25 batches of size 400, so 25 rounds of gradient decent is performed. After training this simple model for 2000 epochs, we obtain a model for predicting closeness centrality of nodes based on the aforementioned local features with an average error of 1.86\% on the test set.
		\item [(ii)] We delete 500 edges randomly from $G$ and add 500 new edges randomly in such a way that the modified graph remains connected. Then we compute the feature set and target set of the new graph exactly with the same method as Part (i). We apply the trained model in Part (i) to the feature set of all 14113 nodes of the modified graph. The average error is 2.195\%. A word of caution is necessary here. After removing and adding 500 edges to graph $G$, the list of starting points of the intervals list computed for graph $G$ was still a good choice for the modified graph. In practice, it is possible that lists of degrees occurring in these graphs are very different and one has to develop a method to choose a common list of starting point for them.
	\end{itemize}
\end{example}

The above example clearly shows the robustness of our method of learning closeness centrality by applying deep learning on $p$-aggregation of the RCDF matrix representations of nodes of the graph. The fact that we only used order 4 RCDF matrix representations of nodes (a locally engineered feature) to learn closeness centrality (a globally defined centrally measure) confirms the local-to-global nature of our method. Moreover, Part (ii) of the above example demonstrates that our method is applicable to dynamic graphs. In Example \ref{exam:clossness-random}, we demonstrate the inductive power of our method to learn closeness centrality. In the following example, we train a model to learn the closeness centrality of the network of facebook pages of media. We also test several alternative ideas for feature engineering. 

\begin{example}
	\label{exam:fb-media-closeness-rcdf-rndfc}
	\begin{itemize}
		\item [(i)] Let $G=(V, E)$ be the graph of facebook pages of media. We use Algorithm \ref{algo:increasing_starting_points} with parameters $s=1$, $m=70$, $d=\max_{v\in V}deg(v)$ and $r=1.3$ to construct an intervals list. For the higher order NDF, we take the order 3 RCDF and RNDFC matrices associated to this intervals list. For $p$-aggregation, we try all vales in the set $\{0.3, 0.2, 0.15\}$. The results of all $p$-aggregations are vectors with 23 entries. We compute the closeness centrality of nodes using NetworkX library for the target set. The dataset is shuffled and divided into 10000 nodes for training and 17917 nodes for evaluating the model. The feedforward neural network model is the same as the model described in Example \ref{exam:fb-co-closeness}. Again we train the model for 2000 epochs. For various options described in the above, the average inaccuracies of the prediction on the test set are as the Table \ref{table:fb-media-errors}. Although the order 3 RNDFC matrix matrices have one row more that the order 3 RCDF matrices, the model using RNDFC matrices performs worse than the model using RCDF matrices. This shows the relative advantage of RCDF matrices for learning closeness centrality. We also notice that the parameter $p$ in the $p$-aggregation is an important hyperparameter for our method. 
		\item [(ii)] To see whether a deeper model can help to lower the average inaccuracy, We consider the order 3 RCDF matrix representation with $p$-aggregation for $p=0.2$ and feed the feature vectors obtained into the model described in Example \ref{exam:fb-co-pagerank}. The average inaccuracy of the trained model is 2.089\%, which is surprisingly higher than the one for shallower model!
	\end{itemize}
\end{example}

\begin{table}
	\begin{tabular}{l|c|c|c}
		& $p=0.3$ & $p=0.2$ & $p=0.15$ \\\hline
		RCDF & 2.122\% & 1.755\% & 1.752\% \\\hline
		RNDFC & 4.398\% & 3.471\% & 2.404\% 
	\end{tabular}
	\caption{\label{table:fb-media-errors} The average inaccuracies of the model described in Example \ref{exam:fb-media-closeness-rcdf-rndfc} on the test set according to different options. Columns are indexed according to parameter $p$ of the aggregation and rows are marked by the matrix representation which was used.}
\end{table}

\begin{example}
	\label{exam:clossness-random} Consider two random graphs discussed in Example \ref{exam:pagerank-random} and the list of starting points we constructed there. We compute the order 2 RCDF matrix w.r.t. the intervals list associated to these starting points and obtain a $2 \times 21$ matrix representation for each node of these graphs. Next, we set $p=0.2$ and apply $p$-aggregation to this matrix representation. The final feature vectors have 21 euclidean dimension. We train the same feedforward neural network model described in Example \ref{exam:fb-co-closeness} with 10000 randomly chosen feature vectors of the first graph as the feature set and the closeness centrality of the corresponding nodes as the target set. The average error of the trained model on the rest of 10000 nodes is 1.42\%. Afterwards, we apply the trained model to the whole dataset of the second graph. The average error of predicting closeness centrality of nodes of the second graph is 2.053. This clearly shows that the ``local'' feature vectors obtained from $p$-aggregation of order 2 RCDF matrix representation of the first graph contains enough statistical information to be used to predict the closeness centrality of not only nodes of the first graph, but also the nodes of the second graph. 	
\end{example}

The above example confirms the inductive nature of our method. The convenience of learning closeness centrality using our methods offers a number of research opportunities on applications of closeness centrality measure in network science which will be developed in our future works. 

Finally, we note that replacing vanilla or minimal intervals lists with a dynamic intervals list induces a row-wise aggregation on the level of higher order NDF matrix representations. We leave the exact formulation of these type of aggregations and the study of their generalizations and possible applications to the interested reader.    


\section{NDF for Directed Graphs}
\label{sec:digraphs}

All of our previous constructions are based on two notions in graphs; the degrees of nodes and the distance between nodes. There are two possibilities for each of these notions in directed graphs; ``inward'' and ``outward'' according to whether the directed edges are towards a node or starting from a node, respectively. Therefore, to extend various constructions of NDF embeddings and its higher order generalizations to directed graphs, we need to consider both of these variations. We do not bore the reader by repeating all aforementioned definitions and constructions according to inward/outward dichotomy and content ourselves with few instances. We begin with the crucial difficulties that arise while dealing with directed graphs. The directed graph shown in Figure \ref{fig:main-graph-directed} is used to demonstrate our reasoning, constructions and basic examples, so we briefly denote it by $G_0=(V_0, E_0)$ in this section.

\begin{figure}
	\centering
	\includegraphics[width=330px]{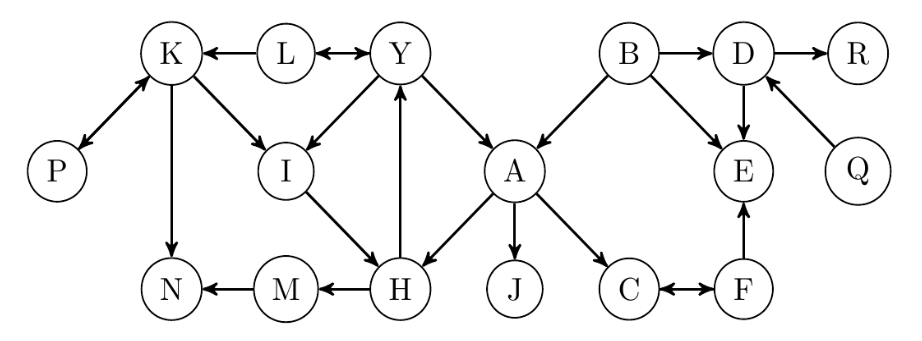}
	\caption{The directed graph illustrating the concepts and constructions of Section \ref{sec:digraphs}.} \label{fig:main-graph-directed}
\end{figure}

Let $G=(V,E)$ be a directed graph. The main issue is that edges in directed graphs are not symmetric. This has several major consequences that should be addressed: 
\begin{itemize}
	\item Given a node $v\in V$, a node $u\in V$ is called an {\bf in-neighbor} (resp. {\bf out-neighbor}) of $v$ if there is a directed edge $(u, v)$ (resp. $(v, u)$). The first consequence of asymmetric edges in directed graphs is that the number of in-neighbors of $v$ (also called {\bf in-degree} of $v$ and denoted by $deg^{in}(v)$) is not always equal to the number of out-neighbors of $v$ (also called {\bf out-degree} of $v$ and denoted by $deg^{out}(v)$). Therefore, all degree frequency calculations must be done accordingly in two parallel ways. We use the superscripts $^{in}$ and $^{out}$ for denoting which way (inward or outward) is intended. For example, the inward vanilla NDF (i.e. the vanilla in-neighbors in-degree frequency) vector of a node $v\in V$ is denoted by $vndf^{in}(v)$. 
	\item The second consequence is the possibility for a node to have in-neighbors (resp. out-neighbors) whose in-degrees (resp. out-degrees) are zero. For example in graph $G_0$, node $B$ is an in-neighbor of node $E$ whose in-degree is zero. This requires us to consider an entry for zero in NDF vectors. Therefore any (vanilla or dynamic) list of starting points should begin with zero instead of 1. When we consider a minimal intervals list for a static graph the list of starting points is allowed to start with a number bigger than 0, but for dynamic version of minimal intervals list, the list of starting point must begin with 0 as well. 
	\item Moreover, we have two maximum degrees in directed graphs; maximum in-degree and maximum out-degree, so construction of a dynamic NDF for a directed graph starts with choosing two lists of starting points accordingly. For graph $G_0$, both the (inward/outward) vanilla (and minimal) lists of starting points are equal to $\{0, 1, 2, 3\}$, the in-degrees list and the out-degrees list of $G_0$. The vanilla inward/outward NDF vectors of some selected nodes of $G_0$ are listed below:
	\begin{eqnarray*}
		vndf^{in}(A) &=& (1, 0, 1, 0), \qquad vndf^{out}(A) = (1, 1, 1, 0) \\
		vndf^{in}(H) &=& (0, 0, 2, 0), \qquad vndf^{out}(H) = (0, 1, 0, 1) \\
		vndf^{in}(B) &=& (0, 0, 0, 0), \qquad vndf^{out}(B) = (1, 0, 1, 1) \\
		vndf^{in}(K) &=& (0, 2, 0, 0), \qquad vndf^{out}(A) = (1, 2, 0, 0) 
	\end{eqnarray*}
	\item Finally, we note that the existence of a directed path $(v_0, \cdots, v_n)$ starting from a node $v_0\in V$  and ending to another node $v_n\in V$ does not guarantee that there is a directed path from $v_n$ to $v_0$. For example, there is a directed path $(B, A, C, F)$ from node $B$ to node $F$ in graph $G_0$, but not vice versa. Moreover, if there is such a path, there is no guarantee that the shortest path from $v_0$ to $v_n$ has the same length as the shortest path from $v_n$ to $v_0$. For example, the shortest path from $K$ to $Y$, i.e. $(K, I, H, Y)$, has length 3 while the one from $Y$ to  $K$, i.e. $(Y, L, K)$, has length 2. This means we have no well-defined symmetric distance (metric) function on the set of nodes of a directed graph. However, we can perform the inward/outward Breath-First Search and define in-circles as well as out-circles around each node. Specifically, the {\bf in-circle} of radius $n$ centered at a node $v\in V$ is defined and denoted as 
	\[
	C_n^{in} (v) = \{u\in V, \text{there is a shoetest directed path from } u \text{ to } v \text{ of length } n\}, 
	\] 
	for $n\geq 0$. The {\bf out-circle} of radius $n$ centered at $v$ is defined similarly and denoted by $C_n^{out} (v)$. For example in $G_0$, we have
	\[
	C_2^{in} (A) = \{L, H\},  \qquad C_2^{out} (A) = \{M, Y, F\}.
	\]
\end{itemize}

Having the notions of (inward/outward) circles in our arsenal, we are able to define all variations of the higher order NDF matrix representations of nodes in a directed graph as straightforward generalizations of these notions in undirected graphs. Therefore we are ready to apply our methods in previous sections to directed graphs as well. We are particularly interested in learning PageRank of directed graphs. 


\subsection{Learning PageRanks of directed graphs}
\label{subsec:pagerank_directed_graphs}
\qquad\\

The original definition of PageRank and our experiments in Section \ref{sec:pagerank} suggest that we should employ the inward NDFC matrix representations as our feature engineering method. Before applying our deep learning methods for learning PageRank, we would like to introduce an intuitive adjustment to NDFC matrix representation as a secondary option for feature engineering necessary for learning PageRank. Then we shall use both versions of NDFC matrices in several examples. Hopefully this adjustment inspires further ideas for advanced feature engineering methods based on higher order NDF representations. 

The PageRank was originally introduced for hyperlink analysis of webpages in the internet, see \cite{BP}. As an iterative algorithm the PageRank of a node $v$ in the directed graph of webpages (with hyperlinks as directed edges) is defined as
\[
PR_{n+1}(v) = \sum_{u\in C_1^{in}(v)} \frac{PR_n(u)}{deg^{out}(u)},
\]
where $PR_n(x)$ is the PageRank of a node $x$ in the $n$-th iteration. The denominator of the summand takes into account the adverse effect of out-degrees of in-neighbors of $v$ on its PageRank. A similar phenomenon can be observed while building the NDFC matrix representation of nodes of a directed graph. For example, assume graph $G_0$ consists of webpages and consider the second row of the inward NDFC matrix of node $H$ in $G_0$ w.r.t. vanilla intervals list. We have $C_1^{in}(H)=\{A, I\}$, and so the second row of the matrix $ndfc^{in}(H)$ is the mean of the inward vanilla NDF vectors of $A$ and $I$, that is 
\[
\frac{(1,0,1,0) + (0,0,2,0)}{2}.
\] 
One observes that $I$ has only one out-neighbor which is $H$ and does not lower the chance of $H$ to be visited after $I$. But $A$ has 3 out-neighbors $\{H, J, C\}$, and besides $H$, the other two out-neighbors of $A$ may attract some of the visitors of $A$ and consequently lower the chance of $H$ to be visited after $A$. Therefore, giving the same weight to the NDF vectors of $A$ and $I$ does not seem right, so we need to discount the inward NDF vectors of in-neighbors of $H$ proportionate to their out-degrees. One obvious solution would be to divide the inward NDF vectors of in-neighbors of $H$ by their out-degrees, just like the iterative formula for PageRank mentioned above.

On the other hand, we note that discounting the NDF vectors must be applied for all circles of positive radius. So besides the first row of the NDFC matrix, the rest of rows will be adjusted by the discounting process. Regarding these considerations, we propose the following discounted version of the inward NDFC matrix representation:

\begin{definition}
	\label{def:discout-function}
	Let $G=(V, E)$ be a directed graph. The {\bf order $r$ discounted inward NDFC matrix representation} with respect to an intervals list $\mathcal{I}=\{I_1,I_2,\cdots,I_m\}$ assigns an $(r+1)\times m$ matrix $ndfc_r^{in\_dis}(v)$ to each node $v\in V$ such that its first row is the same as the first row of $ndfc_r^{in}(v)$ and, for $k>1$, its $k$-th row is defined by	
	\begin{equation*}
		\frac{1}{s_{k-1}^{in}(v)} \sum_{u\in C_{k-1}^{in}(v)} \frac{dndf^{in}(u)}{deg^{out}(u)}, \qquad \text{when } s_{k-1}^{in}(v)\neq 0,
	\end{equation*}
	and the zero vector of dimension $m$ when $s_{k-1}^{in}(v)=0$.
\end{definition}

When $u$ is an in-neighbor of some node, we have $deg^{out}(u) \geq 1$, so the denominator of the summand in the above definition is non-zero. One also notes that considering the mapping $u\mapsto \frac{1}{deg^{out}(u)}$ as the {\bf discount function} is only one simple proposal and finding the best definition for the discount function depends on the downstream application and it requires further research. In the following examples, we use both versions (discounted and not discounted) of the inward NDFC matrix representation as feature vectors to train deep learning models for learning PageRanks of directed graphs.

\begin{example}
	\label{exam:wiki-vote-pagerank}
	Let $G=(V,E)$ be the directed graph of ``Wikipedia vote network'' \cite{LHK}, in which Wikipedia pages are considered as nodes and votes between them are considered as directed edges. This graph has 7115 nodes. To build a list of starting points according to the in-degrees list of $G$, we apply the method explained in Remark \ref{rem:costume-sp} and choose the following list of starting points:
	\[
	\{0, 1, 2, 4, 7, 12, 19, 30, 47, 72, 110, 167, 237, 307, 450\}
	\]
	Then we compute the order 5 inward NDFC matrix representation of nodes of $G$ w.r.t. this list of starting points. After flattening these matrices into one dimensional arrays, we obtain our feature vectors each with 90 entries. As usual, we compute the PageRank of nodes using the Python package NetworkX as the target variable. To balance target variables and feature vectors, we multiply all PageRanks of nodes by 5000. The dataset is shuffled and divided into train/test sets, 5000 nodes for training and the rest of 2115 nodes for evaluating the model success. The rest of details of the training process is as in Example \ref{exam:fb-co-pagerank} except that we choose the learning rate 0.0005 to avoid overfitting. The average inaccuracy of the prediction of PageRank on the test set is 13.09\%. Afterwards, we repeat this experiment with the order 5 discounted inward NDFC matrices and obtain the average inaccuracy rate of 13.87\% on the test set. Due to various randomness in the training process, there is a few percentage points fluctuation around the average inaccuracy rate. Therefore, we can conclude that both versions of NDFC matrices perform almost similarly on this graph. 
\end{example}

In the following example, we try our methods on a larger directed graph:

\begin{example}
	\label{exam:epinions-pagerank}
	Let $G=(V,E)$ be the directed graph of ``Epinions social network'' \cite{RAD}, in which the nodes are the members of the site epinions.com and directed edges denote who trust whom. This graph has 75879 nodes and its maximum in-degree is 3035. Again we apply the method explained in Remark \ref{rem:costume-sp} to build a list of starting points according to the in-degrees list of $G$. The result is the following list of starting points:
	\begin{eqnarray*}
		&&\{ 0, 1, 2, 3, 4, 5, 7, 9, 11, 14, 17, 21, 25, 30, 36, 43, 51, 60, 70, 82, 96, 112, 130, 151, 175, 203, \\
		&&\quad 235, 272, 315, 365, 422, 488, 558, 628, 698, 768, 838, 908, 978, 1100, 1200, 1450, 2800 \}
	\end{eqnarray*}
	
	Then we compute the order 3 inward NDFC matrix representation of nodes of $G$ w.r.t. this list of starting points, which yields matrices of size $4\times 43$. The feature vectors obtained from flattening these matrices have 172 entries. As usual, we compute the PageRanks of nodes using the Python package NetworkX as the target variables. To balance target variables and feature vectors, we multiply all PageRanks of nodes by 10000. The dataset is shuffled and divided into train/test sets, 20000 nodes for training and the rest of 55879 nodes for evaluating the model success. The rest of details of the training process is as in Example \ref{exam:fb-co-pagerank} except that we train this model only for 1000 epochs to avoid overfitting. The average inaccuracy of the prediction of PageRank on the test set is 27.196\%. Afterwards, we repeat this experiment with the order 3 discounted inward NDFC matrices and obtain the average inaccuracy of 19.425\% on the test set. We can attribute this notable improvement to applying discounting in our feature engineering.
\end{example}


\subsection{PageRanks of undirected graphs using discounted NDFC matrix representations}
\label{subsec:pagerank_undirected_discounted}
\qquad\\

Every undirected graph can be considered as directed graph if we consider its symmetric undirected edges as two directed edges opposite to each other. In the directed graph obtained this way all inward and outward concepts and constructions we discussed in this section are equivalent and in particular the discounting process is simply dividing by the degree of the neighbor. That is, given an undirected graph $G=(V,E)$, $v\in V$ and $k>0$, the $k$-th row of $ndfc_r^{dis}(v)$, the {\bf order $r$ discounted NDFC matrix representation} of $v$, is defined by
\begin{equation*}
	\frac{1}{s_{k-1}(v)} \sum_{u\in C_{k-1}(v)} \frac{dndf(u)}{deg(u)}, \qquad \text{when } s_{k-1}(v)\neq 0,
\end{equation*}
and the zero vector when $s_{k-1}(v)=0$. In the following examples we apply the discounted NDFC matrix representations to learn PageRanks of two undirected graphs discussed in Section \ref{sec:pagerank}, i.e. the graph of facebook pages of companies and the graph of facebook pages of media. 

\begin{example}
	\label{exam:fb-co-pagerank-discounted}
	In Example \ref{exam:fb-co-pagerank}, we replace the NDFC matrices with the discounted NDFC matrices and train the model with the same details as before. The trained model has the average inaccuracy of 9.09\% on the test set. We notice that discounted NDFC matrices speed up the training process slightly and during 2000 epochs some overfitting occurs. So we reduce the number of epochs to 1500 and repeat the experiment. This time the trained model has even the lower average inaccuracy of 8.069\%.
\end{example}

The above example clearly suggests that the discounted NDFC matrix representation is slightly more suitable for learning PageRanks of undirected graphs. However using the discounted NDFC matrices does not show an improvement in the following example:

\begin{example}
	\label{exam:fb-media-pagerank-discounted}
	In Example \ref{exam:fb-media-pagerank-ndfc}, we replace the NDFC matrices with the discounted NDFC matrices. The average inaccuracy of trained model on the test set is 13.143\%.
\end{example}

Our conclusion is that the discounting the NDF vectors according to the out-degrees of nodes has a mixed effect in improving the quality of feature vectors for learning PageRank. A comprehensive research is needed to determine exact conditions under which this technique can produce better results. 
 
\section{Conclusion and Final Remarks}
\label{sec:conclusion}

In this article, we introduced the following constructions based on local data gathered from the neighborhoods of nodes of graphs:
\begin{itemize}
	\item A general family of parametric centrality measures.
	\item Static and dynamic versions of NDF vector representations of nodes.
	\item Higher order NDF matrix representations of nodes.
	\item An effective aggregation method on higher order NDF matrix representations.
	\item A discounting process for engineering more efficient feature vectors.
\end{itemize}

Using these constructions, we proposed the following methods:

\begin{itemize}
	\item An approximation of the neighborhoods of nodes by trees of height 2.
	\item A method for graph isomorphism testing based on NDF vectors and RNDFC matrices.
	\item An easy way to visualize the neighborhoods of nodes as color maps.
	\item An inductive deep learning algorithm for learning PageRank which is applicable to dynamic graphs.
	\item An inductive deep learning algorithm for learning closeness centrality which is applicable to dynamic graphs.
\end{itemize}

The above list of applications is not exhaustive and we expect our graph embeddings and machine leaning methods find more applications in other subjects related to machine learning and data science on graphs, network science, graph isomorphism testing, graph labeling, etc. 

\noindent{\large {\bf Final remarks. }} There are several important points that need to be highlighted here: 

\begin{itemize}
	\item Our focus in this article was to define NDF vector embeddings and their higher order matrix representations and other feature engineering methods. We did not pay a through attention to the architecture of deep learning algorithms we have used. Therefore, we think interested readers will find a number of untouched research ideas concerning various designs of deep learning algorithms according to the applications mentioned here or new applications. For example, possibility of learning other centrality measures using NDF embeddings.
	\item As we already said multiple times, our way of training depends on a number of random procedures such as parameter initialization, data shuffling and splitting, etc. Therefore, one should expect minor differences while verifying the performance of our methods. 
	\item The complete set of codes for all of constructions and examples in this article will be available in the author's GitHub account. We hope the curious reader can perform new experiments with our code and gain more insights. 
	\item Definitely, the most important choice to be made in our methods is choosing the right list of starting points according to various statistics of the graph and the task at hand. We strongly encourage readers interested in our methods to pay special attention to this issue as it is in the heart of NDF embeddings.
	\item An interesting research idea is to feed the NDF vectors or their higher order matrix representations into well-known graph representation learning models and graph neural network models as initial state. Due to the competitive performance of NDF and RNDFC embeddings in  graph isomorphism testing, it is expected that those deep learning models gain more expressive power because of the NDF embeddings.
\end{itemize}


\end{document}